%% file: LowMessageComplexityShortcuts.tex
\documentclass[11pt, letter]{article}

\input{preamble}

\begin{document}

\input{abstract}
\input{intro}
\input{prerequisites}
\input{techniques}
\input{twoModels}
\input{randomizedSolutionToStrongModel}

\input{deterministicSolutionToStrongModel}

\newpage
\bibliographystyle{plainnat}
\bibliography{abb,ultimate}

\newpage
\input{appendix}

\end{document}

%% file: preamble.tex
\usepackage{todonotes} 

\usepackage{authblk}
\usepackage{comment}
\usepackage[T1]{fontenc}
\usepackage{multirow}
\usepackage{bigstrut}
\usepackage{amssymb}
\let\proof\relax
\usepackage{amsmath}
\let\endproof\relax
\usepackage{amsthm}
\usepackage{enumerate}
\usepackage[shortlabels]{enumitem}
\usepackage{dsfont}
\usepackage{hyperref}
\usepackage[all]{hypcap}
\usepackage{footnote}
\usepackage{wrapfig}
\usepackage{float}
\usepackage{subcaption}
\usepackage{framed}
\usepackage{algorithm,algorithmicx}
\usepackage[noend]{algpseudocode}
\usepackage{xspace}

\makeatletter
\g@addto@macro{\maketitle}{\@thanks}
\makeatother

\usepackage{thmtools,thm-restate} 

\newcounter{note}[section]

\usepackage{lmodern} 
\usepackage[utf8]{inputenc}
\usepackage[T1]{fontenc}

\usepackage{amsmath}
\usepackage{amsfonts,amsthm}  
\usepackage{fullpage}

\usepackage{mdframed}
\usepackage[numbers,sort]{natbib}

\newcommand{\poly}{\mathrm{poly}} 
\def\congest{CONGEST\xspace}
\def\blockroute{\textsc{BlockRoute}\xspace}
\newcommand{\calg}[1]{\Cref{#1}}

\DeclareMathOperator*{\argmin}{arg\,min}

{
	\begin{enumerate}[label=\emph{\roman*)}]%
	}
	{%
	\end{enumerate}%
}

\newcommand{\FullOrShort}{full}
\ifthenelse{\equal{\FullOrShort}{full}}{
	\newcommand{\fullOnly}[1]{#1}
	\newcommand{\shortOnly}[1]{}
}{

\newcommand{\fullOnly}[1]{}
\newcommand{\shortOnly}[1]{#1}
}

\usepackage{cleveref} 
\theoremstyle{plain}
\newtheorem{thm}{Theorem}[section]
\newtheorem{cor}[thm]{Corollary}

\newtheorem{lem}[thm]{Lemma}

\newtheorem{Def}[thm]{Definition}
\newtheorem{obs}[thm]{Observation}

\crefname{listing}{Program-code}{Program-codes}

%% file: abstract.tex
\pagenumbering{gobble}

\title{Round- and Message-Optimal Distributed Graph Algorithms\footnote{An extended abstract of this work will appear in PODC 2018.}}

\newcommand*\samethanks[1][\value{footnote}]{\footnotemark[#1]}

\author[ ]{Bernhard Haeupler\thanks{Supported in part by NSF grants CCF-1527110, CCF-1618280 and NSF CAREER award CCF-1750808.}, D. Ellis Hershkowitz\samethanks, David Wajc\samethanks}
\affil[ ]{Carnegie Mellon University\\
{\{haeupler,dhershko,dwajc\}@cs.cmu.edu}}
\date{}
\maketitle

\date{}

\renewcommand{\thefootnote}{\arabic{footnote}}
\setcounter{footnote}{0}

\maketitle
\begin{abstract}\normalsize
Distributed graph algorithms that separately optimize for either the number of rounds used or the total number of messages sent have been studied extensively. However, algorithms simultaneously efficient with respect to both measures have been elusive. For example, only very recently was it shown that for Minimum Spanning Tree (MST), an optimal message and round complexity is achievable (up to polylog terms) by a single algorithm in the \congest{} model of communication.

In this paper we provide algorithms that are simultaneously round- and message-optimal for a number of well-studied distributed optimization problems. Our main result is such a distributed algorithm for the fundamental primitive of computing simple functions over each part of a graph partition. From this algorithm we derive round- and message-optimal algorithms for multiple problems, including MST, Approximate Min-Cut and Approximate Single Source Shortest Paths, among others. On general graphs all of our algorithms achieve worst-case optimal $\tilde{O}(D+\sqrt n)$ round complexity and $\tilde{O}(m)$ message complexity. Furthermore, our algorithms require an optimal $\tilde{O}(D)$ rounds and $\tilde{O}(n)$ messages on planar, genus-bounded, treewidth-bounded and pathwidth-bounded graphs.\footnote{Throughout this paper, $n$, $m$ and $D$ denote respectively the number of nodes, number of edges and the graph diameter respectively. In addition, we use tilde notation, $\tilde{O},\tilde{\Omega}$ and $\tilde{\Theta}$, to suppress polylogarithmic terms in $n$.} 
\end{abstract}

\newpage

\pagenumbering{arabic}
\setcounter{footnote}{0}

%% file: intro.tex
\section{Introduction}
Over the years, a great deal of research has focused on characterizing the optimal runtime for distributed graph algorithms in the \congest model of communication. Fundamental problems that have been studied include Shortest Paths \cite{nanongkai2014distributed,frischknecht2012networks,holzer2012optimal,lenzen2013efficient,lenzen2015fast,izumi2014time}, MST \cite{garay1998sublinear,kutten1995fast,peleg2000near,khan2006fast}, Min-Cut \cite{ghaffari2013distributed,nanongkai2014almost}, and Max Flow \cite{ghaffari2015near}. Runtime is measured by the number of synchronous rounds of communication, and for these problems $\tilde{\Theta}(D+\sqrt n)$ rounds are known to be necessary and sufficient \cite{dassarma2012distributed,ghaffari2013distributed,peleg2000near,elkin2006unconditional}.

Another common performance metric optimized for in the \congest model is the total number of messages sent. For MST, an $\tilde{\Omega}(m)$ lower bound is known \cite{awerbuch1990trade}.\footnote{Strictly speaking, the $\tilde{\Omega}(m)$ message lower bound for MST only holds if the algorithm is (1) deterministic (2) in the KT0 model, or (3) ``comparison-based". (Our deterministic algorithm satisfies (1),(2) and (3).) If these conditions are not met it is possible to solve MST using $\tilde{O}(n)$ messages -- beating the $\tilde{\Omega}(m)$ bound for sufficiently dense graphs. For more see \citet{mashreghi2017time}.} However, for several decades the only MST algorithms known to match this message lower bound had  sub-optimal round complexity \cite{gallager1983distributed,chin1985almost,gafni1985improvements,awerbuch1987optimal,faloutsos2004linear,awerbuch1990trade}.
The question of whether algorithms attaining both optimal round and message complexity has been a long-standing problem. For instance, \citet{peleg2000near} asked whether it might be achievable for MST.
In a recent breakthrough work \citet{pandurangan2017time} answered this question in the affirmative, providing a randomized MST algorithm with simulataneously optimal round and message complexities (up to polylog terms). Shortly thereafter \citet{elkin2017simple} provided the same result without randomization. However, simultaneously round- and message-optimal algorithms for other fundamental problems have remained elusive.



\subsection{Our Main Result}\label{sec:results}
In this paper we advance the study of simultaneously round- and message-optimal distributed algorithms. In particular, we provide such algorithms for multiple distributed graph problems.
Underlying these contributions is our main result -- a round- and message-optimal algorithm for a fundamental distributed problem, which we refer to as Part-Wise Aggregation (or PA for short). We elaborate on some applications of this algorithm in \Cref{sec:apps}, as well as \Cref{sec:appApps}. Informally, Part-Wise Aggregation is the problem of computing the result of a function applied to each part of a graph partition.
Formally, the problem is as follows.
\begin{Def}[Part-Wise Aggregation]
	\label{def:pwAgg}
	The input to \emph{Part-Wise Aggregation (PA)} is: 

\begin{enumerate}
\item a graph $G=(V,E)$;
\item a partition $(P_i)_{i=1}^N$ of $V$, where each $P_i$ induces a connected subgraph on $G$. Each node $v \in P_i$ knows an $O(\log n)$-bit value associated with it, $\text{val(v)}$, and which of its neighbors are in $P_i$;
\item a function $f$ that takes as input two $O(\log n)$-bit inputs, outputs an $O(\log n)$-bit output and is commutative and associative.
\end{enumerate}
	
\noindent The problem is solved if for every $P_i$ every $v\in P_i$ knows its part's aggregate value $f(P_i) := f(\text{val}(v_1),f(\text{val}(v_2),\dots))$, where $P_i=\{v_1,v_2,\dots\}$.
\end{Def}


The performance of our algorithm is determined by the quality of the \emph{shortcuts} that the input graph admits. Shortcuts, as well as the parameters which determine their quality -- termed the \emph{block parameter}, $b$, and \emph{congestion}, $c$ -- are formally defined in \fullOnly{\Cref{subsec:shortcuts}}\shortOnly{\Cref{subsec:prelims}}. For now, we note only that every graph admits a shortcut with $b=1$ and $c=\sqrt n$. Our main result is as follows.

\begin{restatable}{thm}{PASoln}
\label{thm:partComputation}
There exists a Part-Wise Aggregation on a graph $G$ admitting a tree-restricted shortcut with congestion $c$ and block parameter $b$ w.h.p.\footnote{Throughout the paper, by w.h.p. we mean with probability $1 - \frac{1}{\poly(n)}$.} in $\tilde{O}(bD+c)$ rounds and $\tilde{O}(m)$ messages, and deterministically in $\tilde{O}(b(D+c))$ rounds and $\tilde{O}(m)$ messages.
\end{restatable}


\subsection{Applications of Our Main Result}\label{sec:apps}
The power of Part-Wise Aggregation -- and by extension \Cref{thm:partComputation} -- is that numerous distributed primitives can be cast as instances of this problem. For example, it is not hard to see that electing a leader, computing the number of nodes in each tree in a forest or having every part of a graph partition agree on a minimum value are all instances of this problem. Consequently, many previous algorithms rely on subroutines which are implementable using Part-Wise Aggregation 
\cite{dassarma2012distributed, 
elkin2006faster,
gallager1983distributed,
garay1998sublinear,
ghaffari2013distributed, 
ghaffari2014near, 
ghaffari2015near,
ghaffari2016distributed, 
haeupler2017beating, 
kutten1995fast,
nanongkai2014almost}.
Perhaps unsurprisingly then, 
using our new PA algorithm as a subroutine 
in some of these previous works' 
algorithms, we obtain round- and 
message-optimal solutions to numerous problems.

In the following three corollaries we highlight 
three such applications of our algorithm: round- and message-optimal algorithms for MST, 
Approximate Min-Cut and Approximate SSSP. We give proofs of these corollaries and also discuss further applications of our PA algorithm and our subroutines in \Cref{sec:appApps}.
For a flavor of these 
corollaries' proofs, we note that Bor\r{u}vka's 
MST algorithm \cite{nesteril2001otakar} 
can be implemented easily using $O(\log n)$ applications of Part-Wise Aggregation, implying \Cref{cor:mst}. Corollaries \ref{cor:minCut} and \ref{cor:sp} are obtained by using our PA algorithm in the algorithms of \citet{ghaffari2016distributed} and \citet{haeupler2017beating}, respectively. 

\gdef\probDfnAll{The input to all three problems consists of an undirected weighted graph, with edge weights in $[1, \poly(n)]$. Initially, every node knows  the weight associated with each of its incident edges. }

\fullOnly{\probDfnAll Since every graph admits a shortcut with $b=1$ and $c = \sqrt{n}$, our algorithms simultaneously achieve message complexities of $\tilde{O}(m)$ and runtimes of essentially worst-case optimal $\tilde{O}(D  + \sqrt{n})$.
	}

\gdef\probDfnMST{\paragraph{MST.}\label{sec:mst} MST is solved when every node knows which of its incident edges are in the MST.}
\fullOnly{\probDfnMST}

\begin{restatable}{cor}{mstcor}\label{cor:mst}
	Given a graph $G$ admitting a tree-restricted shortcut with congestion $c$ and block parameter $b$, one can solve MST w.h.p in $\tilde{O}(bD+c)$ rounds and $\tilde{O}(m)$ messages and deterministically in $\tilde{O}(b(D+c))$ rounds with $\tilde{O}(m)$ messages.
\end{restatable}
\gdef\mstProof{
\begin{proof}
	Our algorithm uses \Cref{thm:partComputation} to simulate Bor\r{u}vka's classic MST algorithm \cite{nesteril2001otakar}. In Bor\r{u}vka's algorithm every node initially belongs to its own part. Then, for $O(\log n)$ rounds each part merges with the part it is connected to by its minimum outbound edge. 
	The problem of determining the minimum-weight outbound edge of a part is an example of Part-Wise Aggregation, which we solve with the algorithm of \Cref{thm:partComputation}. Whenever a node $v$ is incident to this edge it remembers that the neighbor along that edge is now in the same part as $v$. Round and message complexities are trivial and correctness follows from that of Bor\r{u}vka's algorithm.
\end{proof}
}

\gdef\probDfnMC{\paragraph{Approximate Min-Cut.}\label{sec:mincut}
Min-cut is $(1 + \epsilon)$-approximated when every node knows whether or not it belongs to a set $S \subset V$ such that the size of the cut given by $(S, V\setminus S)$ is at most $(1 + \epsilon )\lambda$, where $\lambda$ is the size of the minimum cut of $G$ with the prescribed weights.}
\fullOnly{\probDfnMC}

\begin{restatable}{cor}{mincutcor}\label{cor:minCut}
For any $\epsilon > 0$ and graph $G$ admitting a tree-restricted shortcut with congestion $c$ and block parameter $b$, one can $(1 + \epsilon)$-approximate min-cut w.h.p. in $\tilde{O}\left(bD+c)\cdot \poly(1/\epsilon\right)$ rounds and $\tilde{O}\left(m \right) \cdot \poly(1/\epsilon)$ messages.
\end{restatable}

\gdef\mincutProof{
\begin{proof}
Section 5.2 of \Citet{ghaffari2016distributed} provides an algorithm for approximate min-cut based on shortcuts. The algorithm works roughly as follows: use sampling ideas from \citet{karger1994random} to downsample edge weights so that the min-cut is of size $O(\log n / \epsilon^2)$; by \citet{thorup2007fully}, we can now solve MST $O(\log n) \cdot \poly(1/\epsilon)$ times (using certain different weights each time) such that there is one edge $e^*$ in one of our MSTs $T^*$ such that the two connected components of $T^* \setminus e^*$ define an approximately optimal min-cut; lastly, using a sketching approach this edge can be found by solving PA $\poly \log (n) \cdot \poly(1/\epsilon)$ times with high probability. See \cite{ghaffari2016distributed} for a proof of correctness. Our claimed round and message complexities are trivial given \Cref{cor:mst} and \Cref{thm:partComputation}, as the above algorithm requires downsampling edge weights (only requiring $O(1)$ rounds and $O(m)$ messages) and by \Cref{cor:mst} and \Cref{thm:partComputation}, the $\poly \log n \cdot \poly(1/\epsilon)$ instances of MST and PA can be solved with $\tilde{O}(bD + c) \cdot \poly(1/\epsilon)$ rounds and $\tilde{O}(m)\cdot \poly(1/\epsilon)$ messages.
\end{proof}
}

\gdef\probDfnSSSP{\paragraph{Approximate SSSP.}\label{sec:sssp}
An instance of $\alpha$-approximate single source shortest paths (SSSP) consists of an undirected weighted graph $G$ as above and a source node $s$, which knows that it is the source node. For $v \in V$ we denote by $d(s,v)$ the shortest path length between $s$ and $v$ in $G$ and $L = \max_{u,v}d(u, v)$. The problem is solved once every node $v$ knows $d_v$ such that $d(s, v) \leq d_v \leq \alpha \cdot d(s, v)$. }
\fullOnly{\probDfnSSSP}

\begin{restatable}{cor}{spcor}\label{cor:sp}
	For any $\beta=O(1/\poly\log n)$, given a graph $G$ admitting a tree-restricted shortcut with congestion $c$ and block parameter $b$, one can $L^{O(\log \log n)/ \log(1/\beta)}$-approximate SSSP w.h.p in $\tilde{O}\left(\frac{1}{\beta}(bD + c) \right)$ rounds and $\tilde{O}\left(\frac{m}{\beta}\right)$ messages.
\end{restatable}

\gdef\spProof{
\begin{proof}
	\Citet{haeupler2017beating} provide a distributed algorithm with the stated round complexity and approximation factor based on a solution to PA. Roughly the algorithm works as follows. We compute $O(\frac{\log n}{\beta})$ low diameter decomposition; in particular, as in \Citet{miller2013parallel} every node starts a weighted BFS at a randomly-chosen time and runs this weighted BFS for $O(\frac{\log n}{\beta})$ rounds. During this weighted BFS, nodes claim as part of their ball any nodes they reach that have not yet been claimed or started their weighted BFS. The weights used for the weighted BFS change in each round: any weight strictly inside a claimed ball is updated to have weight 0 and any edge incident to two claimed balls is additively increased. Moreover, the increments used by the weighted BFS geometrically increase in each iteration so that despite increasing edge weights, the weighted BFS can still efficiently proceed. Lastly, the union of all BFS trees returned by every weighted BFS is returned as a tree, $T^*$, that approximates distance in the graph; to inform nodes of their approximate distance from the source, the source can simply broadcast on $T^*$.
	
	What makes this algorithm difficult to implement is that running a weighted BFS with edge weights set to 0 requires that a weighted BFS traverse components connected by weight-zero edges of potentially large diameter ``in a single round''. To overcome this issue, \Citet{haeupler2017beating} use PA to efficiently traverse these connected components. For more see \Citet{haeupler2017beating}.
	Relying on our algorithm for PA and observing that these dominate the round and message complexity of the weighted BFS calls, our claimed round and message complexities follow.
\end{proof}
}

The value of $\beta$ determines a tradeoff between the quality of the SSSP approximation and the round and message complexity of our algorithm. Taking $\beta=\log^{-\Theta(1/\epsilon)} n$, \Cref{cor:sp} yields an $O(L^\epsilon)$-approximation algorithm using $\tilde{O}(bD + c)$ rounds and $\tilde{O}(m)$ messages. 

\subsection{Discussion of Our Results}\label{sec:disc}
There are a few salient points worth noting regarding our results, on which we elaborate below. 

\paragraph{Round- and Message-Optimality of our Algorithms.} 
As all graphs admit a tree-restricted shortcut with block parameter $b=1$ and congestion $c=\sqrt n$, our
algorithms all terminate within 
$\tilde{O}(D+\sqrt n)$ rounds, which is optimal for 
all our applications of our PA algorithm, by
\cite{dassarma2012distributed}. As for message complexity, our
$\tilde{O}(m)$ bound is tight for MST 
in the $KT_0$ model by 
\cite{awerbuch1990trade}. For the other problems an $\Omega(n)$ message lower bound is trivial; for sparse graphs, then, our message complexity bound is tight for these problems.
Finally, we note that our proof of \Cref{cor:mst} relies on solving Part-Wise Aggregation $O(\log n)$ times to solve MST, which implies that our algorithms for PA are both round- and message-optimal (again, up to polylog terms). 


\paragraph{Beyond Worst-Case Optimality.} As stated above, every graph admits tree-restricted shortcuts
with block parameter $b=1$ and congestion $c=\sqrt{n}$, which implies an $\tilde{O}(D+\sqrt{n})$ bound for our algorithms' round complexity on general graphs. However, as observed in prior work, a number of graph families of interest -- planar, genus-bounded, bounded-treewidth and bounded-pathwidth graphs -- admit shortcuts with better parameters \cite{ghaffari2016distributed,haeupler2016near,haeupler2016low}. As a result, our algorithms have a round complexity of only $\tilde{O}(D)$ times the relevant parameter of interest (e.g., genus, treewidth or pathwidth). 
Provided these parameters are constant or even polylogarithmic, our algorithms run in $\tilde{O}(D)$ rounds. Another recent result \cite{haeupler2018minor} implies that our algorithms run in $\tilde{O}(D^2)$ time on minor-free graphs. We elaborate on our results for all the above graphs in \Cref{sec:fullResultsTable}. We also note that our algorithms need not know the optimal values of block parameter and congestion, as a simple doubling trick can be used to approximate the best values (see \cite{haeupler2016low}). In particular, our algorithms perform as well as the parameters of the best shortcut that the input graph admits. 

\paragraph{Future Applications of This Work.}
Non-trivial shortcuts likely exist for graph families beyond those mentioned above. 
As such, demonstrating even better runtimes for our algorithms on many networks may be achieved 
in the future by simply proving the \emph{existence} of efficient shortcuts on said networks. Moreover, 
given the pervasiveness of PA in distributed graph algorithms, the applications of our 
PA algorithm we present are likely non-exhaustive. We are hopeful that our 
PA algorithm will find applications in deriving round- and message-optimal bounds for even more 
problems.

%% file: prerequisites.tex
\section{Preliminaries}\label{subsec:prelims}
Before moving onto our formal results, we explicitly state the model of communication we consider and then review relevant concepts from previous work in low-congestion shortcuts. 
\fullOnly{\subsection{CONGEST Model of Communication}}

Throughout this paper we work in the classic \congest{} model of communication \cite{peleg2000distributed}. In this model, the network is modeled as a graph $G=(V,E)$ of diameter $D$ with $n=|V|$ nodes and $m=|E|$ edges. Communication is conducted over discrete, synchronous rounds. During each round each node can send an $O(\log n)$-bit message along each of its incident edges. Every node has an arbitrary and unique ID of $O(\log n)$ bits, first only known to itself (this is the $KT_0$ model of \citet{awerbuch1990trade}).

\fullOnly{\subsection{Shortcuts and Tree-Restricted Shortcuts}
\label{subsec:shortcuts}}
Low-congestion shortcuts were originally introduced by \citet{ghaffari2016distributed} to solve PA. These shortcuts allow high-diameter parts to communicate efficiently, by using edges outside of parts; this effectively decreases the diameter of the parts. \citet{ghaffari2016distributed} showed how, given a simple low-congestion shortcut, PA can be solved in an optimal number of rounds -- i.e. $\tilde{O}( D + \sqrt{n})$ -- w.h.p. Formally, a low congestion shortcut is defined as follows.


\begin{Def}[Low-Congestion Shortcuts \cite{ghaffari2016distributed}]
Let $G = (V, E)$ be a graph and $(P_i)_{i=1}^N$ be a partition of $G$'s vertex set. $\mathcal{H} = H_1, \ldots, H_N$ where $H_i \subseteq E$ is a \emph{$c$-congestion shortcut with dilation $d$} with respect to $(P_i)_{i=1}^N$ if it satisfies
\begin{enumerate}
\item Each edge $e \in E$ belongs to at most $c$ of the $H_i$.
\item The diameter of $(P_i \cup V(H_i), E[P_i] \cup H_i)$ for any $i$ is at most $d$.\footnote{Here $V(H_i)$ denotes all endpoints of edges in $H_i$ and $E[P_i]$ denote the edges of $G$ with both endpoints in $P_i$.}
\end{enumerate}
\end{Def}

\citet{ghaffari2016distributed} also showed how to compute near-optimal $\tilde{O}(D)$-congestion and $\tilde{O}(D)$-dilation shortcuts for planar graphs, given an embedding of such a graph. This allowed them to obtain $\tilde{O}(D)$-round MST algorithms for this problem, among other results. However, it was not until the work of \citet*{haeupler2016low} that it was demonstrated that  shortcuts could be efficiently computed in general. This work showed that high quality instances of a certain type of shortcut -- \emph{tree-restricted shortcuts} -- can be efficiently approximated. These types of shortcuts are defined as follows. 


\begin{Def}[Tree-Restricted Shortcuts \cite{haeupler2016low}]
Let $G = (V, E)$ be a graph and $(P_i)_{i=1}^N$ be a partition of $G$'s vertex set. A shortcut $\mathcal{H}=(H_i)_{i=1}^N$ is a \emph{$T$-restricted shortcut} with respect to $(P_i)_{i=1}^N$ if there exists a rooted spanning tree $T$ of $G$ with $H_i\subseteq E[T]$  for all $i \in [N]$.
\end{Def}
Since a rooted BFS tree has minimal depth, and the $\tilde{O}(D)$-round $\tilde{O}(m)$-message deterministic leader election algorithm of \citet{kutten2015complexity} allows us to compute a BFS tree in the same bounds, throughout this paper $T$ will be a rooted BFS tree. The same work that introduced tree-restricted shortcuts also introduced a convenient alternative to dilation, termed \emph{block parameter}.

\begin{Def}[Block Parameter \cite{haeupler2016low}]
Let $\mathcal{H} = (H_i)_{i=1}^n$ be a $T$-restricted shortcut on the graph $G = (V, E)$ with respect to parts $(P_i)_{i=1}^n$. For any part $P_i$, we call the connected components of $(P_i \cup V(H_i), H_i)$ the \emph{blocks} of $P_i$, and the number of blocks of $P_i$ its \emph{block parameter}. The \emph{block parameter} of $\mathcal{H}$, $b$, is the maximum block parameter of any part $P_i$.
\end{Def}
 
As shown in \cite{haeupler2016low}, if $T$ is a depth-$D$ tree, the dilation of a $T$-restricted shortcut with block parameter $b$ is at most $O(bD)$. As such, block parameter is a convenient alternative to dilation. See \Cref{fig:shortcutEg} for an example of a $T$-restricted shortcut.

\begin{figure}[h]
	\centering
	\includegraphics[scale=0.18, clip]{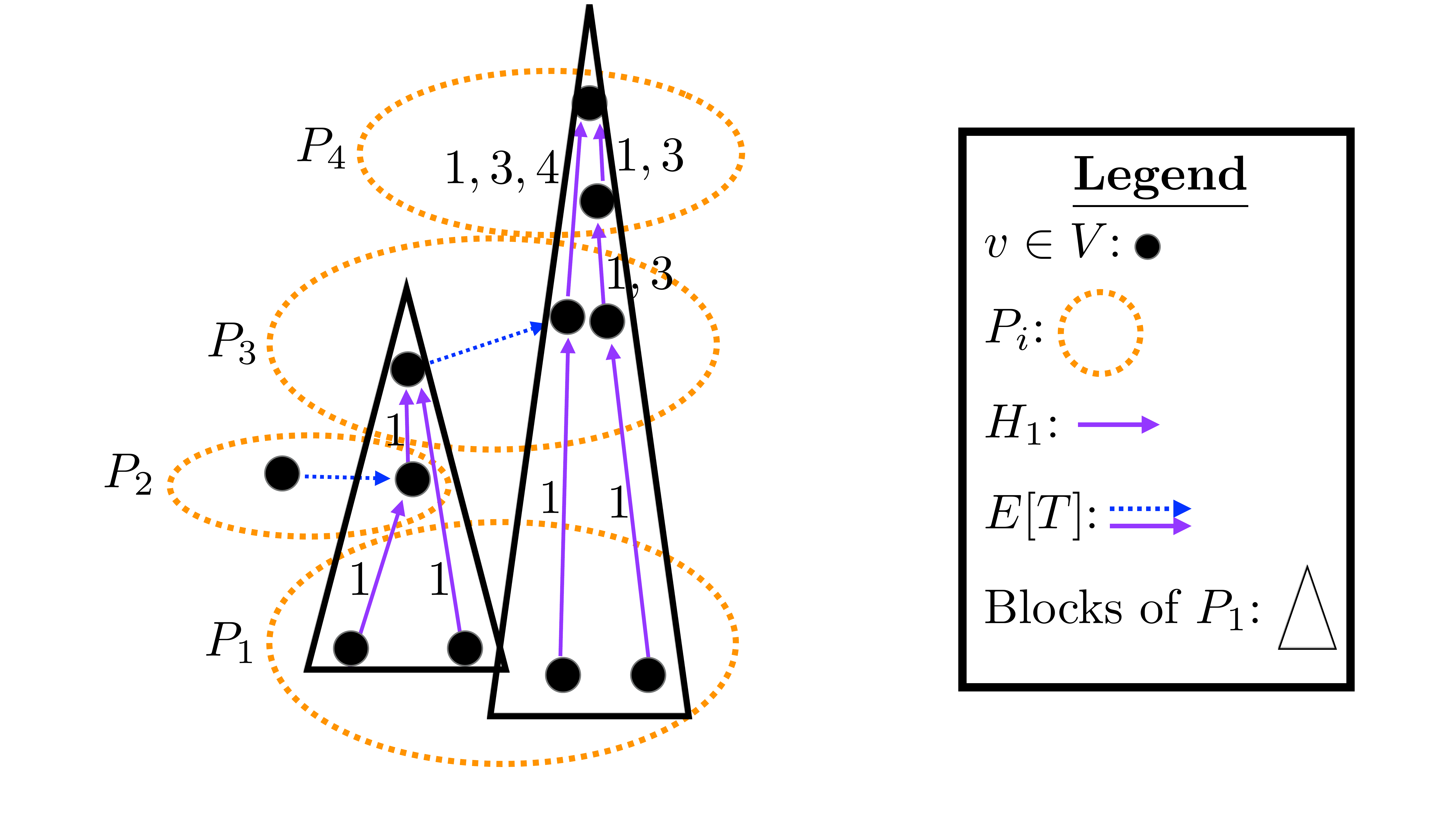}
	\caption{An example of a $T$-restricted shortcut on 4 parts. $e \in E[T]$ labeled with $\{i : e \in H_i\}$. Edges directed towards root of $T$. Here $c=3$ and $b=2$.}
\label{fig:shortcutEg}
\end{figure}

%% file: techniques.tex
\section{Techniques}
In this section we outline our general algorithmic approach. We begin by demonstrating the message sub-optimality of previous shortcut algorithms for Part-Wise Aggregation on a particular example. We then give a workaround for this example and sketch how we develop this workaround into a full-fledged algorithm.

\subsection{Bad Example for Previous Shortcut-Based Algorithms}\label{subsec:messageInOptSCut}

Several prior round-optimal randomized algorithms for PA used tree-restricted shortcuts  \cite{haeupler2016low,haeupler2016near}. To solve PA, these algorithms repeatedly aggregate within blocks. To aggregate within a block, every node in the block transmits its value up the block (along the tree's edges); when values from the same part arrive at a node in the block, they are aggregated by applying $f$ and then forwarded up the block as a single value. By the end of this process, the root of the block has computed $f$ of the block and can broadcast the result back down. This approach can be implemented using an optimal $\tilde{O}(D+\sqrt n)$ rounds.

Unfortunately, there exist PA instances for which the above approach requires $\omega(m)$ messages. For example, consider the $D \times (n-1)/D$ grid graph with an additional node, $r$, neighboring all of the top row's nodes. 
Suppose each row is its own part, and all the column edges are shortcut edges, forming a single block rooted at $r$. 
See \Cref{fig:bad-message-complexity}. Aggregating within this block requires $\Omega(nD)$ messages: 
a message cannot be combined with other messages in its part until it has at least reached $r$ and so each node is responsible for sending a unique message to $r$ along a path of length $D/2$ on average. 
Thus, aggregating in blocks in this way to solve PA requires $\Omega(nD)$ messages, which is sub-optimal for any $D=\omega(1)$, since $m = O(n)$ for this network.

\begin{figure}[t]
	\centering
	\begin{subfigure}{.4\textwidth}
		\centering
		\includegraphics[scale=0.3]{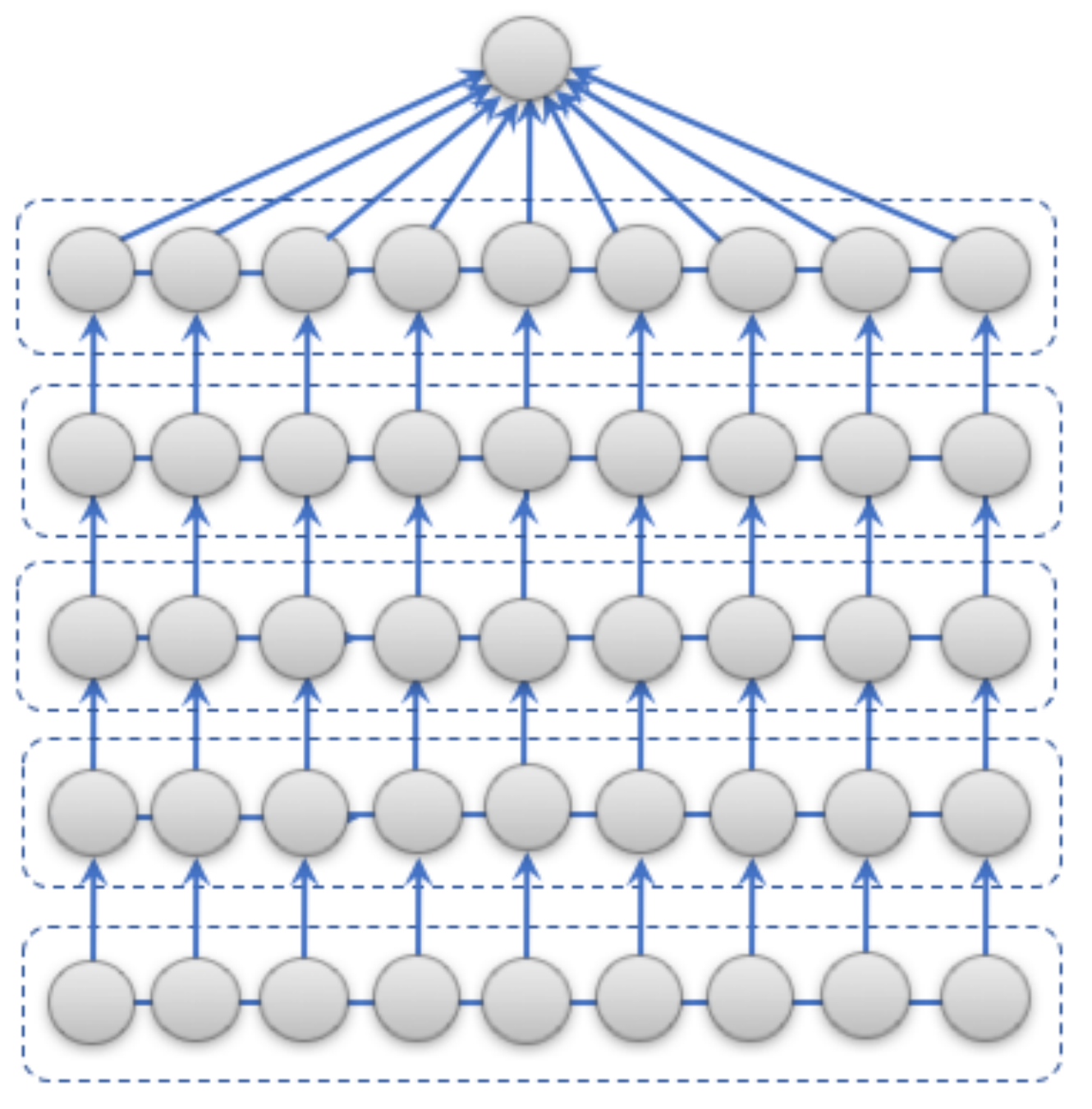}
		\caption{\small $\Omega(nD)$ message example. \\Parts are indicated by dashed rectangles. Tree edges are directed towards the root.}
		\label{fig:bad-message-complexity}
	\end{subfigure}%
	\quad 
	\begin{subfigure}{.4\textwidth}
		\centering
		\includegraphics[scale=0.3]{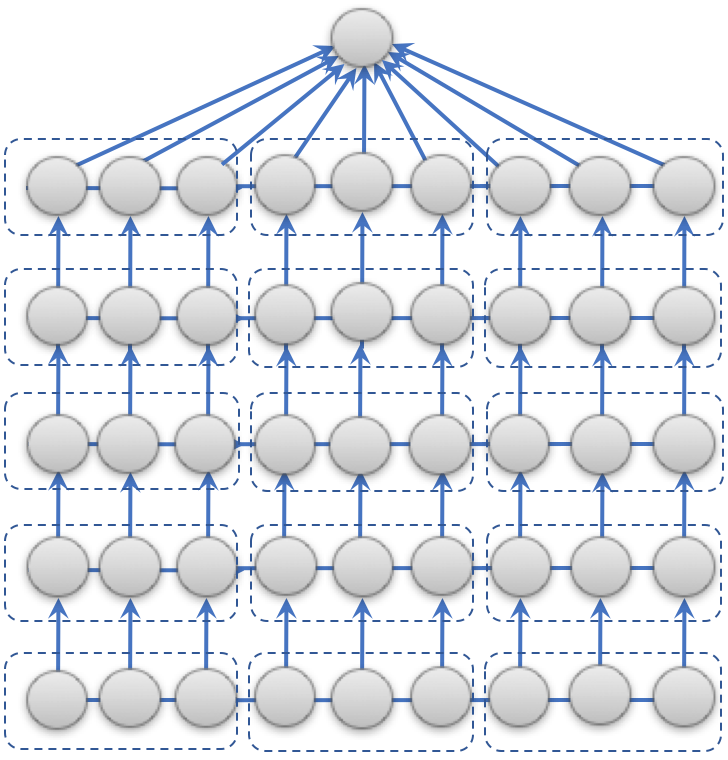}
		\caption{\small A workaround. \\Here sub-parts are indicated by (smaller) dotted rectangles.}
		\label{fig:workaround}
	\end{subfigure}
	\vspace{-0.15cm}
	\caption{A bad example for prior shortcut algorithms, and a workaround.}
	\label{fig:example}
\end{figure}

\paragraph{A Workaround.} We can improve the poor message complexity of aggregating within blocks on this particular network as follows. Partition each of the $D$ parts into \emph{sub-parts}, each with $O(D)$ connected nodes; we have $O(n/D)$ sub-parts in total. See \Cref{fig:workaround}. First, sub-parts aggregate: the right-most node in the sub-part broadcasts its value left and every other node broadcasts left the aggregation of its own value and what it receives from its neighbor to the right. The leftmost node of a sub-part then uses the block's edges to transmit the sub-part's aggregate value to $r$, which then computes the aggregate value for each part. Symmetrically to the above procedure, $r$ then broadcasts to every node the aggregate value for its part.

Aggregating within each sub-part requires $O(n)$ messages, as it requires each node to broadcast at most once. Moreover, there are $O(n/D)$ sub-parts, each responsible for broadcasting up and down the block once and so using the shortcut requires $O(n/D)\cdot O (D) = O(n)$ messages. Therefore, for this network, our workaround requires an optimal $O(m)=O(n)$ messages.

\subsection{Overview of Our Approach}
The workaround of the previous subsection is heavily tailored to the particular example of \Cref{fig:bad-message-complexity}. Moreover, it requires that nodes know significantly more about the network topology than we allow. However, the above example and workaround motivate and highlight some of the notable strategies of our algorithm for Part-Wise Aggregation.


\paragraph{Sub-Part Divisions.}As illustrated in the example, having all nodes use a shortcut in order to send their private information to their part leader rapidly exhausts our $\tilde{O}(m)$ message budget. To solve this issue, we refine the partition of our network into what we call a \emph{sub-part division}. In a sub-part division each part $P_i$ containing more than $D$ nodes is partitioned into $\tilde{O}(|P_i|/D)$ sub-parts each with a spanning tree rooted at a designated node termed the \emph{representative} of the sub-part. In the preceding example the representatives are the left-most nodes of each sub-part. Each sub-part uses its spanning tree to aggregate towards its representative, who then alone is allowed to use shortcut edges to forward the result toward the part leader. This decreases the number of nodes that use the shortcut from $O(n)$ to $\tilde{O}(n/D)$, thereby reducing the message complexity of aggregating within a block from $O(nD)$ to $\tilde{O}(n)$. Applying this observation and some straightforward random sampling ideas to previous work on low-congestion shortcuts to solve PA almost immediately implies our message-efficient randomized solutions to PA.

\paragraph{Message-Efficient (and Deterministic) Shortcut Construction.} If our algorithms are to use shortcuts as we did in the preceding example, they must construct them message efficiently; i.e., with $\tilde{O}(m)$ messages. No previous shortcut construction algorithm achieves low message complexity. We show that not only do sub-part divisions allow us to use shortcuts message efficiently, but they also allow us to construct shortcuts message efficiently. In particular, we give both randomized and deterministic message-efficient shortcut construction algorithms. The latter is the first round-optimal deterministic shortcut construction algorithm and is based on a divide-and-conquer strategy that uses heavy path decompositions \cite{sleator1983data}. Though the general structure of our deterministic shortcut construction algorithm is similar to that used in previous low-congestion shortcut work -- nodes try to greedily claim the shortcut edges they get to use -- the techniques used to deterministically implement this structure are entirely novel with respect to past work in low-congestion shortcuts.

\paragraph{Star Joinings.}
To use sub-part divisions as above, we must demonstrate how to compute them within our bounds. To do so, we begin with every node in its own sub-part and repeatedly merge sub-parts until the resulting sub-parts are sufficiently large.
However, it is not clear how many sub-parts can be efficiently merged together at once, as obtained sub-parts can have arbitrarily large diameter, rendering communication within a sub-part  infeasible. We overcome this issue by always forcing sub-parts to merge in a star-like fashion; this limits the diameter of the new sub-part, enabling the new sub-part to adopt the representative of the center of the star.
We call this behavior \emph{star joinings}. As we show, enforcing this behavior is easily accomplished with random coin flips. We also accomplish the same behavior deterministically but with significantly more technical overhead, drawing on the coloring algorithm of \citet{cole1986deterministic}.

%% file: twoModels.tex
\section{Solving PA}
\label{sec:techSet}
In this section we show how to solve PA, given shortcuts and a sub-part division.
The subroutines necessary to compute shortcuts and sub-part divisions randomly and deterministically within our time and message bounds are given in \Cref{sec:rando} and \Cref{sec:det}, respectively. Those subroutines together with our algorithm for PA given a sub-part division and shortcuts imply our main result, \Cref{thm:partComputation}. \shortOnly{Due to space constraints we only outline the algorithms here, and sketch their proofs of correctness and analysis, deferring formal pseudocode and proofs to \Cref{sec:appKLPA}.}

For our purposes it is convenient to assume that in our PA instance each part $P_i$ also has a \emph{leader} $l_i\in P_i$ where every $v \in P_i$ knows the ID of $l_i$. As we show in \Cref{sec:PASolns}, we can dispense with this assumption at the cost of logarithmic overhead in round and message complexity. As we ignore multiplicative polylogarithmic terms, for the remainder of the paper we assume that a leader for each part is always known in our PA instances. 

One of the crucial ingredients we will rely on to solve PA instances as above is sub-part divisions.

\begin{Def}[Sub-part division]
	Given partition $(P_i)_{i=1}^N$ of $V$, a \emph{sub-part division} is a partition of every part $P_i$ into $\tilde{O}\left(\frac{|P_i|}{D} \right)$ \emph{sub-parts} $S_1, \ldots, S_{k_i}$. Each sub-part $S_j$ also has a spanning tree of diameter $O(D)$ rooted at a node $r \in S_j$, termed the sub-part's \emph{representative}.
\end{Def}

\gdef\SPAndBlocksFig{
\begin{figure}[h]
	\centering
	\vspace{-0.4cm}
	\includegraphics[scale=0.15, clip]{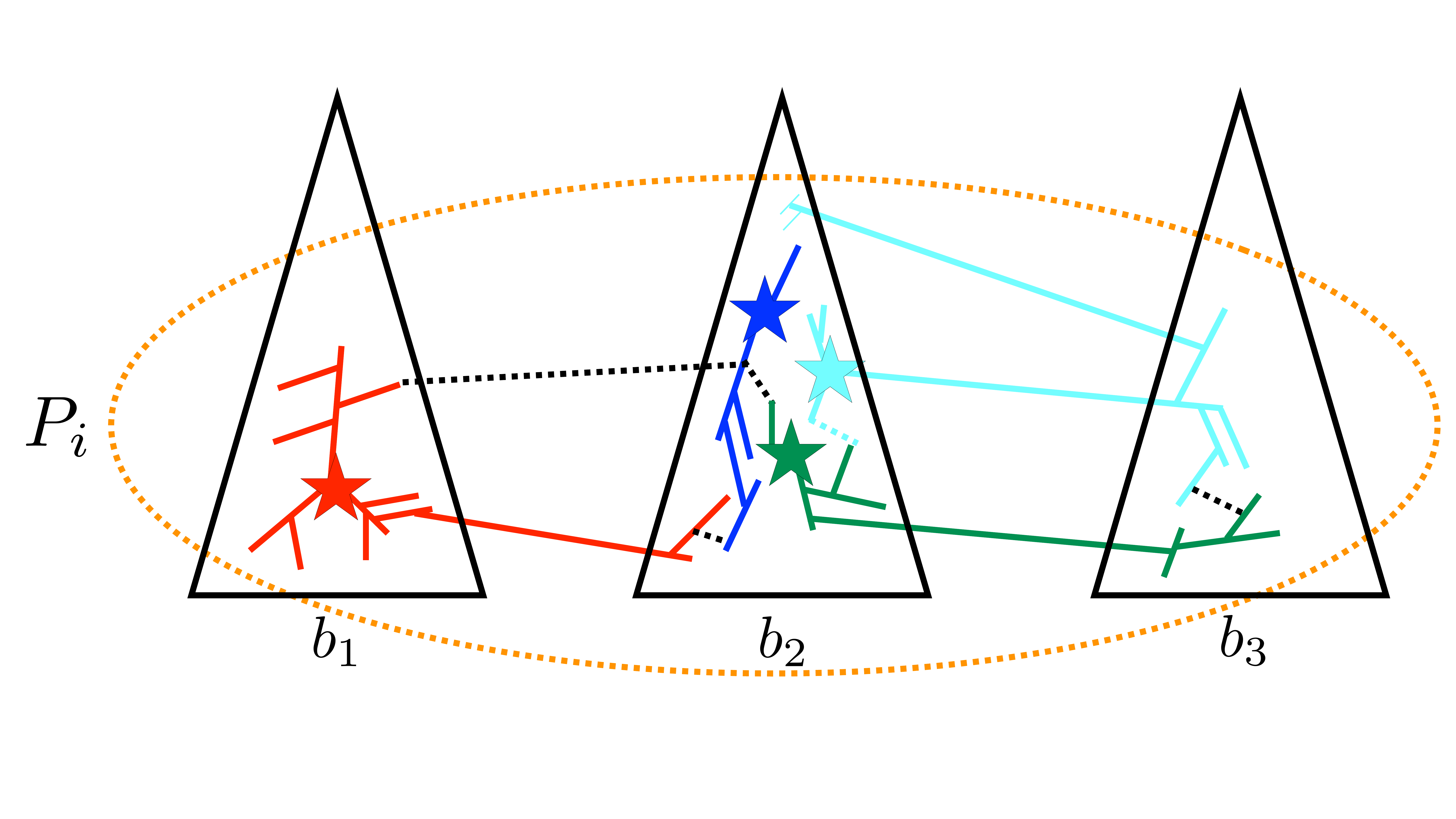}
	\vspace{-.5cm}
	\caption{A division of a part, $P_i$, incident to blocks $b_1$, $b_2$ and $b_3$, into 4 sub-parts. Sub-part representatives: stars. Solid colored lines: edges in the tree of each sub-part according to the color of the  representative. Dashed lines: edges in $E$ between sub-parts.}
\label{fig:subPartAndBlocks}
\end{figure}
}\SPAndBlocksFig

\noindent We note that sub-parts are not necessarily related to blocks in any way; e.g.\ a single sub-part might span multiple blocks and blocks need not contain sub-part representatives. See \Cref{fig:subPartAndBlocks}\fullOnly{ for an illustration of how sub-parts and blocks might interact}.



\fullOnly{\subsection{Aggregating on Families of Sub-trees}}
The second ingredient we rely on is tree-restricted shortcuts, along which we will route (some of) our messages. To do so, we must first restate an algorithm of \citet*{haeupler2016low} which we refer to as \blockroute, that convergecasts/broadcasts within shortcut blocks. As convergecast and broadcast are symmetric, we only discuss convergecast.

\begin{restatable}{lem}{blockRoute}(\cite[Lemma 2]{haeupler2016low})
	\label{lem:subTreeRouting}
Let $T$ be a tree of depth $D$. Given a family of subtrees such that any edge of $T$ is contained in at most $c$ subtrees, there is a deterministic
algorithm that can perform a convergecast/broadcast on all
of the subtrees in $O(D + c)$ \congest{} rounds.

Specifically, for convergecasts, if multiple messages are 
scheduled over the same edge, the algorithm forwards the
packet with the smallest depth of the subtree root, breaking
ties with the smallest ID of the subtree.
\end{restatable}

One observation we make about this algorithm, and which will prove crucial since we only allow representatives to use shortcuts, is the following.
\begin{obs}\label{obs:blockroutemessage}
Let $S$ be the set of nodes with a value to be convergecasted in the algorithm described in \Cref{lem:subTreeRouting}. Then the number of messages used by the algorithm is $O(|S|D)$.
\end{obs}


\gdef\blockRoutePC{
\begin{algorithm}
\begin{algorithmic}[1]
\Statex \textbf{Input}: Blocks given by $B_i(v)$; 
\Statex  \color{white}\textbf{Input}: \color{black} representatives for each part $(R_i)_{i=1}^N$.
	
\Statex \textbf{Output}: $v \in R_i$ learns $f(B_i(v) \cap R_i)$.
\Statex \textbf{Notation}: $\kappa_i(v) = \mathds{1}(v \in R_i) + \sum_{v' \text{ $T$-child of $v$}} \mathds{1}(\text{$v'$ has a descendant in $T$ in $R_i$})$.
\For {$v \in V$}
\For{$i \in [N]$}
\State $S_{i, v} \leftarrow \{v\} \cap R_i$. \Comment{We overload ``$v$'' with $v$'s input value}
\State $O_v \leftarrow \{\}$.\Comment{$v$'s outstanding messages}
\If {$v \in R_i \wedge \kappa_i(v) == 1 \wedge v \text{ not $B_i(v)$ root}$}
\State $O_v \leftarrow O_v \cup \{(f(S_{i, v}, i))\}$.
\EndIf
\EndFor
\For{$O(D + c)$ rounds}
	\If{$|O_v| > 0$}
	\State $\hat{i } \leftarrow \argmin_{i : (f(S_{i,  v}), i) \in O_v} T\text{-depth}(\text{root of } B_i(v))$.
	\State Send $(f(S_{\hat{i}, v}), \hat{i})$ to parent of $v$ in $T$.
	\State $O_v \leftarrow O_v \setminus (f(S_{\hat{i}, v}), \hat{i})$.
	\EndIf
\For {$(o, i)$ received}
	\State $S_{i, v} \leftarrow S_{i, v} \cup \{o\}$.
	\If{$|S_{i, v}| == \kappa_i(v) \wedge \text{$v$ not $B_i(v)$ root}$}
	\State $O_v \leftarrow O_v \cup \{(f(S_{i, v}), i)\}$.
	\EndIf
\EndFor

\EndFor 
\EndFor
\State Each root $r$ of each block $B_i(r)$ broadcasts $f(B_i(r) \cap R_i)$ symmetrically to how it was calculated.
\caption{Aggregation within blocks.}
\label{alg:blockRoute}
\end{algorithmic}
\end{algorithm}
}


%

\fullOnly{\subsection{Solving PA and Verifying the Block Parameter}}
We now show how given a sub-part division and a $T$-restricted shortcut, we can round- and message-efficiently solve PA with and without randomization. Our method is given by \calg{alg:KLPASolve} (which contains both our deterministic and randomized algorithm), and works as follows. First, each leader $l_i$ of part $P_i$ broadcasts an arbitrary message $m_i$ to all nodes in $P_i$. Then, symmetrically to how $m_i$ was broadcast, each $l_i$ computes $f(P_i)$ and then broadcasts $f(P_i)$ to all nodes of $P_i$. The most technically involved aspect of our algorithm is how $l_i$ broadcasts $m_i$ to all nodes in $P_i$. If $|P_i|$ is smaller than $D$, broadcast can be trivially performed along the spanning tree of the single sub-part of $P_i$ in $O(D)$ rounds with $O(|P_i|)$ messages. However, if $|P_i|$ is larger than $D$, we use shortcuts, as follows.

For our deterministic algorithm, we repeat the following $b$ times: every representative in a 
block which received the message $m_i$ spreads $m_i$ to other representatives in its block using \blockroute along the shortcut. Next, 
representatives with $m_i$ spread $m_i$ to nodes in their sub-part. Lastly, nodes with 
$m_i$ spread $m_i$ to neighboring nodes in adjacent sub-parts. Crucially, only our 
representatives use shortcuts, thereby limiting our message complexity, by 
\Cref{obs:blockroutemessage}. We illustrate the broadcast of $m_i$ in 
\Cref{fig:algoIterations}\shortOnly{ in \Cref{sec:appKLPA}}. 

Our randomized algorithm works similarly, with the following modification: each part leader 
independently delays itself -- and subsequently, its entire part -- before sending its first message at the beginning of the algorithm, by a delay chosen uniformly in the 
range $[c]$ (here $c$ is the shortcut's congestion). This limits the number of parts which would use any given edge during any round to $O(\log n)$ w.h.p. As only one message can be sent along an edge, we execute \blockroute as before, but rather than break ties as in \Cref{lem:subTreeRouting}, we simply spend 
$O(\log n)$ rounds between each ``meta-round'', to allow each node to forward 
all of its $O(\log n)$ messages. This broadcast within blocks requires $O(D\log n)$ \congest rounds. 

The following lemma states the performance of 
our algorithms.
\gdef\KLPAWithSPDAndSC{
	
\begin{algorithm}
	\caption{PA given shortcut and sub-part division.}
	\label{alg:KLPASolve}
	\begin{algorithmic}[1]
		\Statex \textbf{Input}: PA instance; 
		\Statex  \color{white}\textbf{Input}: \color{black} 
		sub-part division;
		\Statex  \color{white}\textbf{Input}: \color{black} 
		$T$-restricted shortcut
		\Statex \textbf{Output}: solves PA

			\Statex \textbf{Notation}: $S(v)\subseteq V$:= $v \in V$'s sub-part;
				\Statex \color{white}\textbf{Notation}: \color{black} $r(v)\in V$:= $v \in V$'s representative;
		\Statex  \color{white}\textbf{Notation}: \color{black}
		$S(U) : =\bigcup_{u\in U} S(u)$ for $U\subseteq V$;
		\Statex  \color{white}\textbf{Notation}: \color{black}
		$R(U):=\{r(U) \mid u\in U\}$ for $U\subseteq V$;
		\Statex  \color{white}\textbf{Notation}: \color{black}
		$R_i :=r(P_i)$;
		\Statex  \color{white}\textbf{Notation}: \color{black}
		$B_i(v) \subseteq V:=$ the block of part $P_i$ containing $v \in V$

		\Statex
		\For{Part $P_i$}\label{line:KLPAfirst}
		\If{$|P_i| < D$} 
		\State Broadcast $m_i$ from $l_i$ to all of $P_i$ along $P_i$'s spanning tree.
		\Else
		\If{Randomized algorithm}
			\State Delay part $P_i$ by (independent) $\sim U(c)$;
			\State Blow up subsequent calls to \blockroute by $O(\log n)$.
		\EndIf 
		\State Route $m_i$ from $l_i$ to $r(l_i)$ using \blockroute.  \label{line:routeInit}
		\State $\mathcal{A} \leftarrow \{r(l_i)\}$, $\mathcal{I} \leftarrow \{\}$. \Comment{Initialize sets of ``active''/``inactive'' representatives.}
		\For{$b$ iterations}
		\State Run \blockroute on $\mathcal{A}$ to send $m_i$ to all nodes in $\bigcup_{r \in \mathcal{A}} B_i(r) \cap R_i$. \label{line:br}
		\State $\mathcal{A} \leftarrow \mathcal{A} \bigcup_{r \in \mathcal{A}} B_i(r) \cap R_i$.
		\ForAll{$r \in \mathcal{A}$}
			\State Broadcast $m_i$ from $r$ to $S(r)$ along $S(r)$'s representing tree. \label{line:bcastInSP}
		\EndFor
		\State Broadcast $m_i$ over edges in $E$ that exit sub-parts in $S(\mathcal{A})$. \label{line:bcastAcrossSP}
		\ForAll{Vertex $v\not \in S(\mathcal{A}) \cup S(\mathcal{I})$}
			\If{$v$ received a message in line \ref{line:bcastAcrossSP}}
				\State $v$ routes $m_i$ to $r(v)$. \label{line:route}
			\EndIf
		\EndFor
		\State $\mathcal{I} \leftarrow \mathcal{I} \cup \mathcal{A}$. \label{line:inact}
		\State $\mathcal{A} \leftarrow$ representatives that received a message in line \ref{line:route}.\label{line:KLPAlast}
		\EndFor
		\EndIf
		\EndFor
		\State Symmetrically to lines \ref{line:KLPAfirst}-\ref{line:KLPAlast}, compute $f(P_i)$ at $l_i$. \label{line:computef}
		\State Symmetrically to lines \ref{line:KLPAfirst}-\ref{line:KLPAlast}, broadcast result of $f(P_i)$ from $l_i$ to all nodes in $P_i$. \label{line:bcastf}

	\end{algorithmic}
\end{algorithm}

}
\fullOnly{\KLPAWithSPDAndSC}

\gdef\algoIterationsFig{
\begin{figure}[h]
\centering
	\includegraphics[scale=0.35]{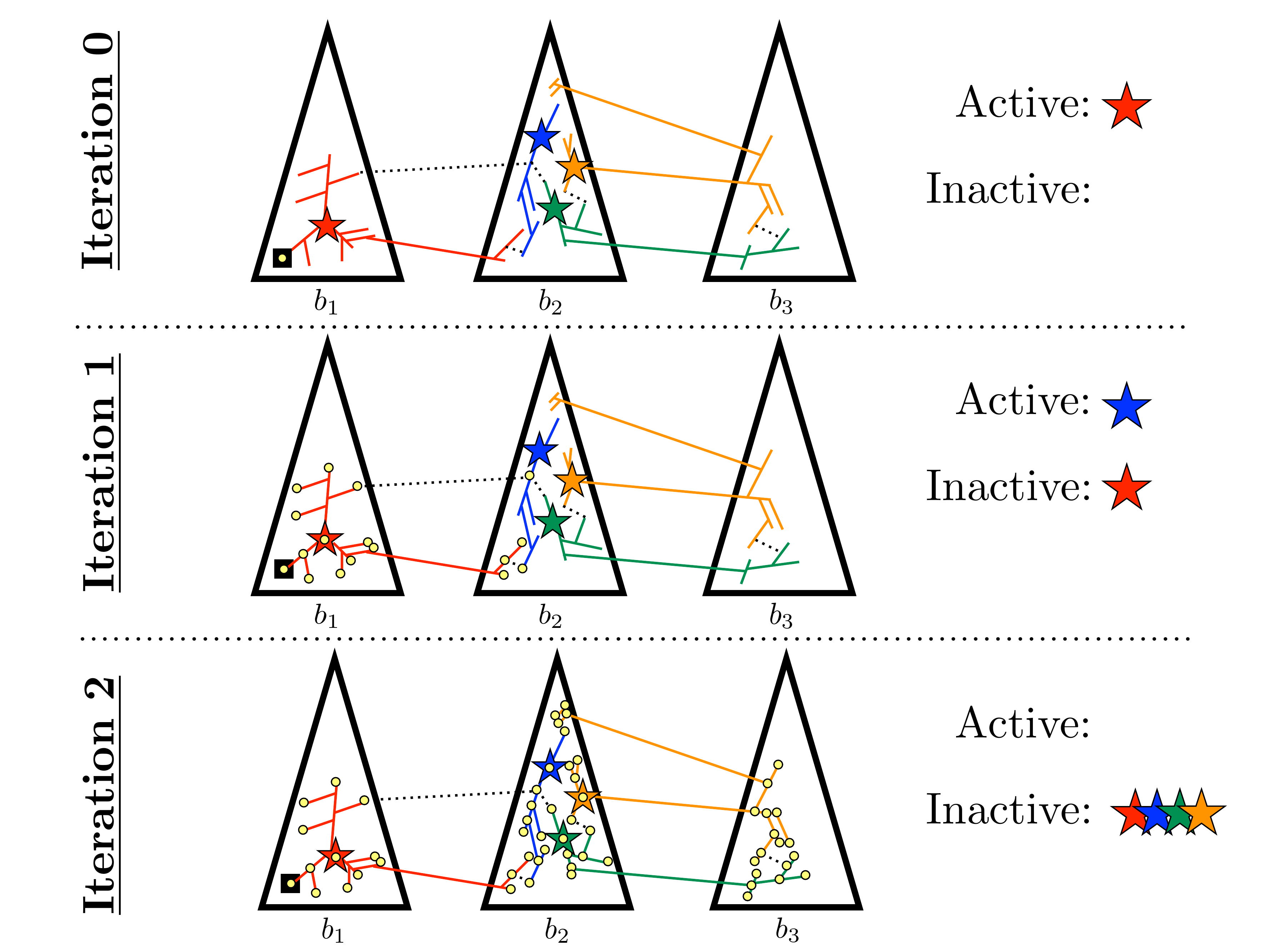}
	\caption{Nodes with $m_j$ (yellow circles) and (in)active representatives at the end of each broadcast iteration for a part with 3 blocks $b_1$, $b_2$ and $b_3$. The leader $l_i$ is indicated by a black square, $l_i$, while sub-part representatives are indicated by stars; solid lines and dotted black lines correspond to intra- and inter-sub-part edges.}
	\label{fig:algoIterations}
\end{figure}
}
\fullOnly{\algoIterationsFig}

\begin{restatable}{lem}{solveKLPA}
\label{lem:SPClusF}
Given a sub-part division and a $T$-restricted shortcut with congestion $c$ and block parameter $b$, \calg{alg:KLPASolve} uses $\tilde{O}(m)$ messages to solve PA either w.h.p in $\tilde{O}(bD + c)$ rounds or deterministically in $\tilde{O}(b(D+c))$ rounds.
\end{restatable}
\gdef\solveKLPAProof{
\begin{proof}
We first prove our round complexities. We start by proving the stated round complexity for broadcasting $m_i$. 
Any part that is of fewer than $D$ nodes clearly only requires $O(D)$ rounds. For any part of more than $D$ nodes, 
we argue that each of the $b$ iterations requires only $\tilde{O}(D+c)$ rounds or $O(D \log n)$ if a random delay of $U(c)$ is added. 
Routing $m_i$ from $l_i$ to $r(l_i)$ requires $O(D)$ rounds. Running \blockroute only requires $O(D+c)$ rounds by \Cref{lem:subTreeRouting}. Moreover, if a random delay is added, 
a Chernoff and union bound show that w.h.p an edge never has more than $O(\log n)$ distinct parts' aggregate messages that should be transmitted along it. 
By allowing each node to send up to $O(\log n)$ parts' aggregate message in each meta-round, \blockroute requires $O(D \log n)$ rounds, and therefore this approach requires $\tilde{O}(bD+c)$ rounds overall. 
Next, broadcasting $m_i$ within any sub-part requires $\tilde{O}(D)$ rounds as sub-parts are of diameter $\tilde{O}(D)$. 
Broadcasting $m_i$ to adjacent subparts requires only a single round. Lastly, computing $f(P_i)$ and broadcasting $f(P_i)$ symmetrically require $\tilde{O}(b(D+C))$ rounds.

We now prove a message complexity of $\tilde{O}(m)$. We start by proving this message complexity for broadcasting $m_i$. 
Message complexity is trivial if the part is of fewer than $D$ nodes. Next consider parts of more than $D$ nodes. 
Notice that nodes in a given sub-part only send messages in those rounds where the sub-part is active. 
Moreover, once a sub-part becomes inactive, it never again becomes active. 
Routing $m_i$ from $l_i$ to $r(l_i)$ requires $O(D)$ messages. 
Moreover, each of the $\tilde{O}\left(\frac{|P_i|}{D}\right)$ sub-parts in part $P_i$ use \blockroute at most once, using $O(D)$ messages per sub-part by \Cref{obs:blockroutemessage}; 
as a result this step uses $\tilde{O}(|P_i|)$ messages for part $P_i$ and so $\tilde{O}(n)$ messages total. 
Broadcasting within all sub-parts requires $O(n)$ messages since each sub-part only does so once and has a spanning tree with which to do so. 
Broadcasting across sub-parts uses each edge at most twice and so uses $O(m)$ messages. 
Lastly, computing $f(P_i)$ and broadcasting $f(P_i)$ symmetrically require $O(m)$ messages.

Correctness of broadcasting $m_i$ is trivial if $|P_i| < D$. 
Moreover, if $|P_i| \geq D$, a simple argument by induction over blocks shows that $b$ iterations suffices for parts of more than $D$ nodes. 
Correctness of computing $f(P_i)$ and broadcasting $f(P_i)$ symmetrically follow.
\end{proof}
}
\fullOnly{\solveKLPAProof}
\shortOnly{
\begin{proof}[Proof (Sketch)]
The deterministic round complexity follows by summing the round complexities of sub-routines. The randomized round complexity results similarly (taking into account the delay of $c$ rounds), but requires noting that a Chernoff and union bound shows that w.h.p.~no more than $O(\log n)$ messages are scheduled over any edge in any given meta round. The only nuance to message complexity is that when \blockroute is run, each of the $\tilde{O}\left(1+\frac{|P_i|}{D}\right)$ sub-parts use $O(D)$ messages, by \Cref{obs:blockroutemessage}. It remains to prove this algorithm's correctness.
Correctness of broadcasting $m_i$ is trivial if $|P_i| < D$. 
Moreover, if $|P_i| > D$, a simple argument by induction over blocks shows that $b$ iterations suffices for parts of more than $D$ nodes. 
Correctness of computing $f(P_i)$ and broadcasting $f(P_i)$ symmetrically follow.
\end{proof}
}

Because our PA algorithm is essentially the same algorithm we use to verify that our shortcuts have good block parameter, we now describe this second algorithm. We verify the block parameter of a fixed part $P_i$ as follows. Run \calg{alg:KLPASolve} to broadcast an arbitrary message $m_i$. If our block parameter is sufficiently small then every node will receive $m_i$ and assume it as such. Moreover, if our block parameter is too large but \calg{alg:KLPASolve} succeeds we can still use \calg{alg:KLPASolve} to inform all nodes in $P_i$ of $P_i$'s block parameter symmetrically to how $m_i$ was broadcast. However, if \calg{alg:KLPASolve} fails --i.e. some node does not receive $m_i$ -- then we must somehow inform all nodes that the block parameter is too large. We do so by having each node that does not receive $m_i$ inform its neighbors in $P_i$ that it did not receive $m_i$. There must be some such neighbor in $P_i$ that did receive $m_i$. By one additional call to \calg{alg:KLPASolve} this neighbor can inform all nodes that did receive $m_i$ that the block parameter is, in fact, too large. This algorithm gives the following lemma.


\gdef\verifyBlockPC{
\begin{algorithm}
	\caption{Block parameter verification.}
	\label{alg:verifyBlock}
	\begin{algorithmic}[1]
	\Statex \textbf{Input}: partition of $V$, $(P_i)_{i=1}^N$, where $v \in P_i$ knows leader $l_i$; 
	\Statex  \color{white}\textbf{Input}: \color{black}
	sub-part division; 
	\Statex  \color{white}\textbf{Input}: \color{black}
	$T$-restricted shortcut; 
	\Statex  \color{white}\textbf{Input}: \color{black}
	desired block parameter $b$ 
	\Statex \textbf{Output}: for every $P_i$, $v \in P_i$ learns if $P_i$ has block parameter $b$ in the input shortcut
	\For{part $P_i$}
	\State Run \calg{alg:KLPASolve} to broadcast arbitrary message $m_i$ from $l_i$.
	\For{$v \in P_i$ that did not receive $m_i$}
		\State $v$ broadcasts $\bar{m}_i$ to neighbors in $P_i$. \label{line:tellNeighSad}
	\EndFor
	\State Run \calg{alg:KLPASolve} to broadcast if a node that received $m_i$ also received $\bar{m}_i$.
	\For{every $i$ and $v \in P _i$}
		\If{$v$ did not receive $m_i$ or received $\bar{m}_i$}
		\State $v$ decides block parameter of $P_i$ exceeds $b$.
		\Else{ Run \calg{alg:KLPASolve} to compute the block number of $P_i$.}
		\EndIf
	\EndFor		
	\EndFor	
	\end{algorithmic}
\end{algorithm}
}\fullOnly{\verifyBlockPC}

\begin{restatable}{lem}{verifyb}
\label{lem:verifyb}
Given parts $(P_i)_{i=1}^N$, a sub-part division, a $c$-congestion $T$-restricted shortcut, $\mathcal{H}$, and desired block parameter $b$, one can deterministically (resp.,  w.h.p.) inform every node whether its part's block parameter in $\mathcal{H}$ exceeds $b$ in $\tilde{O}(b(D+c))$ (resp.  $\tilde{O}(bD + c)$) rounds with $\tilde{O}(m)$ messages.
\end{restatable}
\gdef\verifybProof{
\begin{proof}

Round and message complexities follow trivially from \Cref{lem:SPClusF} and \Cref{lem:subTreeRouting}. We now argue correctness. If a node does not receive $m_i$ when \calg{alg:KLPASolve} is first run then the block parameter of $P_i$ is certainly larger than $b$. When this occurs, all nodes will either be told by $l_i$ that the block parameter is larger than $b$ or they will not receive $m_i$, implicitly informing them that the block parameter of $P_i$ is larger than $b$. If all nodes receive $m_i$, then $l_i$ clearly distributes to all nodes in $P_i$ the number of blocks incident to $P_i$ and so the block number of $P_i$ is correctly determined to be above or below $b$ as desired.
\end{proof}
}
\fullOnly{\verifybProof}
\shortOnly{
}


%% file: randomizedSolutionToStrongModel.tex
\section{Randomized Subroutines}
\label{sec:rando}
In this section we outline how we construct sub-part divisions and shortcuts round- and message-optimally using randomization. \shortOnly{Pseudocode and full proofs are deferred to \Cref{sec:appDet}.} 

\fullOnly{\subsection{Computing Sub-Part Divisions with Randomization}}

We first show how a sub-part division can be computed with randomization, by randomly 
sampling sub-part representatives in \calg{alg:subPartRand}. In particular, for large parts ($|P_i|\geq D$), 
every node decides to be a representative with probability $\min\{1, \frac{\log n}{D}\}$ and 
then representatives claim balls of radius $D$ around them as their sub-part. \calg{alg:subPartRand}'s properties are given below.

\gdef\subPartRandPC{
\begin{algorithm}
	\caption{Randomized sub-part division.}
	\label{alg:subPartRand}
	\begin{algorithmic}[1]
	\Statex \textbf{Input}: partition of $V$ given by $(P_i)_{i=1}^N$ where $v\in P_i$ knows leader $l_i$ 
	\Statex \textbf{Output}: sub-part division
	\For{part $P_i$}
	\If{$|P_i| \leq D$}
		\State Let $P_i$ have one sub-part with representative $l_i$.
		\State Compute sub-part spanning tree by an $O(D)$ round BFS restricted to $P_i$ starting at $l_i$.
	\Else
	\For{$v \in P_i$}
		\State With prob.~$\min \{1, \frac{\log n}{D}\}$, node $v$ is a representative and sends its ID to $P_i$ neighbors.

		\For{$O(D)$ rounds}
		\State $v$ broadcasts the first representative ID it hears to neighbors in $P_i$ once.
		\State $v$'s sub-part parent is the neighbor from which it first heard a representative ID.
		\State $v$ determines for which of its neighbors it is a parent.
		\EndFor
	\EndFor
	\EndIf
	\EndFor
	\end{algorithmic}
\end{algorithm}
}\fullOnly{\subPartRandPC}

\begin{restatable}{lem}{SPDivRand}
\label{lem:subPartDivRand}
\calg{alg:subPartRand} computes a sub-part division of a part with a known leader w.h.p in $O(D)$ rounds with $O(m)$ messages.
\end{restatable}
\gdef\SPDivRandProof{
\begin{proof}

Runtime and message complexity are trivial. Correctness is trivial for parts of fewer than $D$ nodes, so consider parts of more than $D$ nodes. By construction, each claimed sub-part has diameter $O(D)$. It remains to show that every node has a representative and there are $\tilde{O}\left(\frac{|P_i|}{D}\right)$ sub-parts in $P_i$. Fix node $v$ and consider the ball of radius $D$ around $v$. Since $P_i$ has at least $D$ nodes, this ball is of size at least $D$ and so a Chernoff bound shows that w.h.p $\Theta(\log n)$ nodes in this ball will elect themselves a representative, meaning $v$ will have a representative. A union bound over all $v$ shows this to hold for every node. Moreover, the expected number of representatives in part $P_i$ is $\frac{|P_i| \log n}{D}$ and so a Chernoff bound shows that w.h.p there are $\tilde{O}\left(\frac{|P_i|}{D}\right)$ sub-parts in $P_i$. A union bound shows this holds for all parts w.h.p.
\end{proof}
}
\fullOnly{\SPDivRandProof}
\shortOnly{
\begin{proof}[Proof (Sketch)]
Chernoff and union bounds prove the sub-part division is of the appropriate size. This fails for small parts $P_i$ (with $|P_i|<D$), but for these parts it is trivial to compute a sub-part division by having the part's leader be the representative. 
\end{proof}
}

\fullOnly{\subsection{Computing Shortcuts with Randomization}}
We now show in \Cref{alg:SCutRand} how we message-efficiently construct a $T$-restricted shortcut with randomization. We rely on the \textsc{CoreFast} shortcut construction algorithm of \citet{haeupler2016low}. In \textsc{CoreFast}, a sub-sampled set of vertices broadcast up $T$, attempting to ``claim'' edges; edges with too many vertices trying to claim them are discarded. To control the message complexity, we only have the $\tilde{O}(n/D)$ sub-part representatives attempt to claim edges. The correctness and runtime of \calg{alg:SCutRand} is given by \Cref{thm:randSCut}.

\gdef\randSCPC{\begin{algorithm}
	\caption{Randomized shortcut construction.}
	\label{alg:SCutRand}
	\begin{algorithmic}[1]
	\Statex \textbf{Input}: partition of $V$, $(P_i)_{i=1}^N$ where $v\in P_i$ knows leader $l_i$; 
	\Statex  \color{white}\textbf{Input}: \color{black}
	BFS tree $T$; 
	\Statex  \color{white}\textbf{Input}: \color{black} 
	sub-part division 
	\Statex \textbf{Output}: $T$-restricted shortcut with congestion $\tilde{O}(c)$ and block parameter $<3b$
	\State Set all $P_i$ \textit{active}.
	\For{$O(\log n)$ iterations}
		\State Run \textsc{CoreFast} \cite{haeupler2016low} shortcut construction algorithm on representatives in active parts.
		\State Run \Cref{alg:verifyBlock} to compute if block parameter of parts exceed $3b$.
		\State Set every $P_i$ with block parameter $< 3b$ on \textsc{CoreFast} result \emph{inactive}.
		\State Let every newly inactive $P_i$ use the shortcut edges assigned to it by \textsc{CoreFast}.
	\EndFor
	\end{algorithmic}
\end{algorithm}}\fullOnly{\randSCPC}

\begin{restatable}{lem}{SCutRand}
\label{thm:randSCut}
Given partition partition of $V$, $(P_i)_{i=1}^N$ where $v\in P_i$ knows leader $l_i$, a sub-part division, spanning tree $T$ and the existence of a $T$-restricted shortcut with congestion $c$ and block parameter $b$, \calg{alg:SCutRand} computes a $T$-restricted shortcut with congestion at most $\tilde{O}(c)$ and block parameter at most $3b$ in $\tilde{O}(bD + c)$ rounds with $O(n)$ messages w.h.p.
\end{restatable}
\gdef\SCutRandProof{
\begin{proof}
We first argue runtime and message complexity. \citet*[Lemma 4]{haeupler2016low} show that \textsc{CoreFast} takes $O(D\log n + c)$ rounds. However, in this algorithm every node potentially sends a message up $T$ once leading to super-linear message complexity. By amending \textsc{CoreFast} so only the $\tilde{O}\left(\frac{n}{D}\right)$ sub-part representatives send a message up $T$ once as we do, it is easy to see that  the algorithm uses only $\tilde{O}(n)$ messages total. Lastly, \Cref{lem:verifyb} shows that block parameter with randomization uses only $\tilde{O}(bD+c))$ and $\tilde{O}(m)$ messages.

We now argue correctness. \citet*[Lemma 4]{haeupler2016low} show that each time \textsc{CoreFast} is run, it computes a $T$-restricted shortcut with block parameter at most $3b$ for at least half of the nodes and congestion at most $8c$. It is easy to see that only having sub-part representatives participate in \textsc{CoreFast} does not affect correctness and so we conclude that after $O(\log n)$ iterations every $P_i$ has been rendered inactive. By construction every $P_i$ has block parameter $< 3b$ and since the congestion of any edge increases by at most $8c$ in any iteration of \calg{alg:SCutRand}, the total congestion of our returned shortcut is $\tilde{O}(c)$.
 \end{proof}
 }
 \fullOnly{\SCutRandProof}
 \shortOnly{
 \begin{proof}[Proof (Sketch)]
Since only the $\tilde{O}(n/D)$ sub-part representatives broadcast up a tree of depth $O(D)$, each uses $O(D)$ messages and so we use $O(n)$ messages. Round complexity is trivial. Congestion of the shortcut follows by construction. The block parameter can be shown by arguing that we can only discard so many extra edges relative to the optimal shortcut and each extra discarded edge induces a new block. 
\end{proof}
 }

%% file: deterministicSolutionToStrongModel.tex
\section{Deterministic Subroutines}
\label{sec:det}
In this section we show how to construct sub-part divisions and shortcuts deterministically.\shortOnly{ Again, we defer pseudocode and full proofs to  \Cref{sec:appDet}.}
\fullOnly{\subsection{Computing Star Joinings Deterministically}}
\shortOnly{\subsection{Computing Sub-Part Divisions}}
Our algorithm for constructing sub-part divisions repeatedly merges together sub-parts until they are of sufficient size. However, if sub-parts are allowed to merge arbitrarily, the resulting sub-parts may have prohibitively large diameter. The diameter of resulting sub-parts can be limited by forcing sub-parts to always join in a star-like fashion. As such, we begin by providing a deterministic algorithm to enable such behavior. We term this behavior a \emph{star joining}.

\begin{Def}[Star joining]
Let $(P_i)_{i=1}^N$ partition $V$. We say a \emph{star joining} is computed over parts $(P_i)_{i=1}^N$ if the following holds:
a constant fraction of the parts $P_i$ are designated as \emph{receivers}, and the other parts $P_i$ are designated as \emph{joiners}. For every joiner part $P_i$, all $v \in P_i$ knows some (common) edge with one endpoint in $P_i$ and another end-point in some receiver part $P_j$.
\end{Def}

We now show how a star joining can be computed deterministically, given a deterministic PA solution. We use as a sub-routine the 3-coloring algorithm of \citet{cole1986deterministic}. Roughly, the \citet{cole1986deterministic} algorithm works as follows. Every node begins with its ID as its color, meaning there are initially $n$ colors. Next, every node updates its color based on its neighbors' colors, logarithmically reducing the number of possible colors. This is then repeated $\log^* n$ times. For more, see \citet{cole1986deterministic}. The properties of this algorithm are as follows.

\begin{lem}\label{lem:cole-vishkin}(\cite[Corollary 2.1]{cole1986deterministic})
	An oriented $n$-vertex graph with maximum out-degree of one can be $3$-colored in $O(\log^* n)$ rounds with $O(m \log^*n)$ messages.
\end{lem}

\gdef\starJoinPC{
\begin{algorithm}
	\caption{Deterministic star joining.}
	\label{alg:starJoinDet}
	\begin{algorithmic}[1]
	\Statex \textbf{Input}: $(P_i)_{i=1}^N$ s.t. $v \in P_i$ knows edge $e_i$ exiting $P_i$ and leader $l_i$; 
	\Statex  \color{white}\textbf{Input}: \color{black}
	PA algorithm $\mathcal{A}$  
	\Statex \textbf{Output}: a star joining on $(P_i)_{i=1}^N$
	\State $\mathcal{J}, \mathcal{R} \leftarrow \emptyset$ \Comment{Initialize joiners and receiver}
	\State $G' \leftarrow \left(((P_i)_{i=1}^N), \{e_i\}_{i=1}^N \right)$
	\State $\mathcal{R} \leftarrow \mathcal{R} \cup \{P_i : \delta^-_{G'}(P_i) \geq 2\}$ by running $\mathcal{A}$.
	\State $\mathcal{J} \leftarrow \mathcal{J} \cup \{P_i : P_i \not \in \mathcal{R} \wedge e_i = (P_{i'}, P_i) \text{ s.t. } P_{i'} \in \mathcal{R} \}$  by running $\mathcal{A}$.\label{line:starJoinFirstStage}
	\State $G' \leftarrow G' \setminus (\mathcal{R} \cup \mathcal{J})$.
	\State Run the 3-coloring algorithm of \citet{cole1986deterministic} on $G'$ \label{line:simColeVish}.
	\For{color $k=1,2,3$}
		\State{$\mathcal{R} \leftarrow \mathcal{R} \cup \{P_i : P_i \text{ colored } k\}$} by running $\mathcal{A}$.\label{line:starJoinSecStage}
		\State $\mathcal{J} \leftarrow \mathcal{J} \cup \{P_i : P_i \not \in \mathcal{R} \wedge e_i = (P_{i'}, P_i) \text{ s.t. } P_{i'} \in \mathcal{R} \}$ by running $\mathcal{A}$.
	\EndFor
	\end{algorithmic}
\end{algorithm}}\fullOnly{\starJoinPC}

We give our algorithm for deterministically computing star joinings in \calg{alg:starJoinDet} which works as follows. Take the super-graph whose nodes are parts and whose edges are the chosen (directed) edges. First, designate parts with at least two incoming edges receivers and all parts with an outgoing edge into one such part a joiner. These parts constitute all trees in our super-graph and so we next remove them from the super-graph, leaving only (directed) paths and (directed) cycles. On the remaining paths and cycles, simulate the Cole-Vishkin algorithm to compute a 3-coloring of the remaining nodes in the super-graph. For colors $k=1,2,3$ make all $k$-colored parts receivers, their neighbors joiners and remove these parts from this process. The properties of our deterministic star joining algorithm are given by the following lemma.
\begin{restatable}{lem}{sJoin}
\label{lem:starJoinDet}
Let $(P_i)_{i=1}^N$ partition $V$ and suppose every $v \in P_i$ knows some edge $e_i \in E$ exiting $P_i$. If algorithm $\mathcal{A}$ solves PA over $(P_i)_{i=1}^N$, then \calg{alg:starJoinDet} computes a star joining over $(P_i)_{i=1}^N$ with $O(\log^*n)$ calls to $\mathcal{A}$.
\end{restatable}
\shortOnly{
\begin{proof}[Proof (Sketch)]
Standard graph-theoretic arguments show that a constant fraction of nodes are receivers. Moreover, it is not too hard to see that the Cole-Vishkin algorithm can be efficiently simulated within our bounds. 
\end{proof}
}

\gdef\sJoinProof{
\begin{proof}
We begin by proving correctness. 
\Cref{line:starJoinFirstStage} yields stars of joiners centered around receivers. Moreover, the union of all nodes designated in \Cref{line:starJoinFirstStage} from a forest with trees of internal degree at least 2. Therefore, the number of internal (marked) super-nodes (and therefore the number of stars) in \Cref{line:starJoinFirstStage} is at most one half of the super-nodes of the tree. 

Now consider the result of \Cref{line:starJoinSecStage}. As no super-node in $G'$ has in-degree at least two at this point in the algorithm, the super-graph considered in \Cref{line:starJoinSecStage} consists of directed cycles and paths. Thus, each time we remove a $P_i$ from the super-graph we remove at most three super-nodes from the graph and $P_i$ gets to merge with its neighbor. It follows that at least $\frac{1}{3}$ of these super-nodes are merged. Combining the first and second stage, we find that the super-nodes are combined into stars, where the number of obtained nodes is less than 2/3 of the original nodes. Therefore the above algorithm computes a star joining.

We now argue that our algorithm requires $O(\log ^* n)$ runs of $\mathcal{A}$. This clearly holds for all sub-routines of our algorithm except for \Cref{line:simColeVish}. In particular, we must argue how the Cole-Vishkin algorithm can be efficiently simulated on our super-graph using $O(\log ^* n)$ runs of $\mathcal{A}$. We repeat the following $O(\log ^* n)$ times. Let $l_i$ be the known leader of $P_i$. Each $P_i$ begins with the color of $l_i$'s ID. Next, the node in $P_i$ incident to the edge chosen by $P_i$ routes the color it received to $l_i$ using $\mathcal{A}$. Then, $l_i$ performs the Cole-Vishkin computation and then broadcasts $V_i$'s new color to all nodes in $P_i$ using $\mathcal{A}$.
\end{proof}
}
\fullOnly{\sJoinProof}

\fullOnly{\subsection{Computing Sub-Part Divisions Deterministically}}
We now use star joinings to deterministically compute sub-part divisions in \calg{alg:subPartDet} as follows: start with each node in its own sub-part; compute star joinings and merge stars of joiners centered around receivers $O(\log n)$ times, fixing sub-parts once they have at least $D$ nodes. Correctness and runtime of \calg{alg:subPartDet} are given by the following lemma.

\gdef\detSPPC{
\begin{algorithm}
	\caption{Deterministic sub-part division.}
	\label{alg:subPartDet}
	\begin{algorithmic}[1]
	\Statex \textbf{Input}: partition of $V$, $(P_i)_{i=1}^N$
	\Statex \textbf{Output}: a sub-part division
	\For{part $P_i$}
		\State $\mathcal{I}_i \leftarrow \{\{v\} : v \in P_i\}$\Comment{Initialize incomplete sub-parts}
		\State $\mathcal{C}_i \leftarrow \{\}$\Comment{Initialize complete sub-parts}
	\For {$O(\log n)$ iterations }
		\For{$F_j \in \mathcal{I}_i$}
			\If{$\exists$ edge $(u, v)\in F_j \times (\bigcup_{F_{j'}\in \mathcal{I}_i} F_{j'}\setminus F_j)$}
				\State $e_j \leftarrow e$ for such edge $e=(u,v)$ using PA.
			\Else
				\State $e_j \leftarrow e$ for some edge $e=(u,v)\in F_j \times (\bigcup_{F_{j'}\in \mathcal{C}_i} F_{j'})$ using PA.
			\EndIf
		\EndFor

	\State Run \calg{alg:starJoinDet} on $\mathcal{I}_i$ sub-parts with edges $\{e_j\}$ \label{line:detStarJoin} to compute a star-joining.
	\For{Joiner $F_j$ with edge $e_j = (u, v)$ with $v \in F_{j'}$}
		\State $F_j$ merges with $F_{j'}$ using PA. \Comment{if $F_{j'}\in \mathcal{C}_i$, update $\mathcal{C}_i$ accordingly.}
		\State $u$ remembers $v$ as its parent.
		\State $F_j$ orients its tree edges to $v$ using PA.
	\EndFor
	\State $\mathcal{C}_i' \leftarrow \{F_j \in \mathcal{I}_i : |F_j| \geq D\}$ using PA.
	\State $\mathcal{C}_i \leftarrow \mathcal{C}_i \cup \mathcal{C}_i'$.
	\EndFor
	\EndFor
	\State \Return Division given by $\{\mathcal{C}_i\}_{i=1}^N$.

	\end{algorithmic}
\end{algorithm}
}\fullOnly{\detSPPC}

\begin{restatable}{lem}{SPDivDet}
\label{lem:comDiv}
Given partition $(P_i)_{i=1}^N$ of $V$, \calg{alg:subPartDet} computes a sub-part division of $(P_i)_{i=1}^N$ in $\tilde{O}(D)$ rounds with $\tilde{O}(n)$ messages.
\end{restatable}

\gdef\SPDivDetProof{
\begin{proof} 
 

Round and message complexities are trivial apart from the fact that we must show that PA can be solved within our bounds on incomplete sub-parts. However, notice that an incomplete sub-part has fewer than $D$ nodes by definition along with a spanning tree in which every node knows its parent; as such aggregating within each incomplete sub-part is trivially achievable with $O(D)$ rounds and $O(m)$ messages.

We now argue correctness. Correctness for parts of fewer than $D$ nodes is trivial. Consider parts of more than $D$ nodes. Sub-parts continue to merge until they are complete and have at least $D$ nodes and so our division clearly has $\tilde{O}\left(\frac{P_i}{D}\right)$ sub-parts. It remains to show that every complete sub-part's spanning tree has diameter $\tilde{O}(D)$. When a complete sub-part results from two incomplete sub-parts joining, its spanning tree has diameter at most $2D$. Call these nodes the \emph{core} of the complete sub-part. When an incomplete sub-part $F_j$ -- which has spanning tree with diameter at most $D$ since it has fewer than $D$ nodes by definition -- joins a complete sub-part, it necessarily joins by way of nodes in the core. Thus, any node in $F_j$ is within $3D$ of any node in the core by way of the resulting sub-part's spanning tree. Similarly, any other subsequent incomplete sub-part that joins the complete sub-part will be within $4D$ of any nodes in $F_j$ by way of the associated spanning tree. Thus, every complete sub-part has spanning tree with diameter at most $4D$.
\end{proof}
}
\fullOnly{\SPDivDetProof}
\shortOnly{
}

\subsection{Computing Shortcuts  Deterministically}

Having shown how sub-part divisions can be computed in a deterministic fashion, we now turn to our deterministic shortcut construction. We rely on heavy path decompositions~\cite{sleator1983data}. 

\begin{Def}[Heavy Path Decomposition \cite{sleator1983data}]
	Given a directed tree $T$, an edge $(u,v)$ of $T$ is  \emph{heavy} if the number of $v$'s descendants is more than half the number of $u$'s descendants; otherwise, the edge is \emph{light}. A \emph{heavy path decomposition} of $T$ consists of all the heavy edges in $T$.
\end{Def}

It is immediate from the definition that each leaf-to-root path on an  $n$-node tree $T$ intersects at most $\lfloor \log_2 n\rfloor$ different paths of $T$'s heavy path decomposition.
Given a rooted tree $T$ of depth $D$, a heavy path decomposition of $T$ can be easily computed in $O(D)$ rounds using $O(n)$ messages.


Our deterministic shortcut construction algorithm, \calg{alg:SCutDet}, first computes a heavy path decomposition and then computes shortcuts on the obtained paths in a bottom-up order. Thus, we first provide a sub-routine, \calg{alg:pathSCut}, that computes shortcuts of congestion $O(c \log D)$ on a path $P$. 
\gdef\SCPathText{\calg{alg:pathSCut} assumes every node $v$ begins with a set $S(v)$ of part IDs that would like to use $v$'s parent edge in the path. For simplicity, we assume vertices of $\mathcal{P}$ are numbered by their height, $v=1,2,\dots$ (i.e., the source of the path is number $1$, its parent is numbered $2$, etc').}\fullOnly{\SCPathText} \calg{alg:pathSCut} iteratively extends paths used for shortcuts, repeatedly doubling them in length, unless too much congestion results. See \Cref{fig:path-shortcut}. This algorithm's properties are as follows.

\gdef\detSCPC{
\begin{algorithm}
	\caption{Deterministic shortcut construction for paths.}
	\label{alg:pathSCut}
	\begin{algorithmic}[1]
		\Statex \textbf{Input}: Path $P \subseteq V$; 
		\Statex  \color{white}\textbf{Input}: \color{black}
		Mapping $S : V \rightarrow 2^{(P_j)_{j=1}^N}$;
		\Statex  \color{white}\textbf{Input}: \color{black}
		Desired congestion $c$
		\Statex \textbf{Output}: Mapping $S_f : V \rightarrow 2^{(P_j)_{j=1}^N}$
		\For{$v \in V$} 
			\State Set $S_0(v) \leftarrow S(v)$.
		\EndFor
		\For{$i = 0,1,2,\dots, \log_2D-1$}
			\For{every node $v \equiv 2^i \mod 2^{i+1}$}
			\If{$|S_i(v)| \geq 2c$} 
				\State Break $(v, v + 1)$ and set $S_i(v) \leftarrow \emptyset$.

			\Else{ 			\State $u \gets v + 2^i$
			\If{no broken edges between $v$ and $u$}
			\State Transmit $S_i(v)$ from $v$ to $u$ along $P$.
			\State Set $S_{i+1}(u)\leftarrow S_i(u) \cup S_i(v)$.
			\EndIf}
			\EndIf
			\EndFor
		\EndFor\\
		\Return $S_f = S_{\log_2 D}$.
	\end{algorithmic}
	\end{algorithm}
}\fullOnly{\detSCPC}

\begin{restatable}{lem}{SCutPath}
\label{lem:path-shortcut}
	Given directed path $P$ of length $D$, desired congestion $c$ and $S : V \rightarrow 2^{(P_i)_{i=1}^N}$ which denotes for each vertex $v$ which parts want to use $v$'s parent edge in $P$, \calg{alg:pathSCut} returns $S_f : V \rightarrow 2^{(P_i)_{i=1}^N}$ s.t.\ for every $v \in P$ it holds that $|S_f(v)| = O(c \log D)$ in $O(c\log D + D)$ rounds.
\end{restatable}
\gdef\SCutPathProof{
\begin{proof}
	To bound the running time, observe that iteration $i$ of the algorithm can be implemented in $c+2^i$ rounds.  Summing over all iterations 
	$i=0,1,\dots,\log_2 D-1$, the bound on the number of rounds follows. 
	To bound the congestion of the output shortcuts, we prove by induction that before the $i$-th iteration the congestion on any edge is at most $2ci$. This clearly holds for $i=0$. Assume as an inductive hypothesis that before iteration $i$ all edges are used by at most $2ci$ parts. In iteration $i$ the only edges whose congestion are potentially increased are those edges exiting $u$ (i.e.\ edge $(u, u+1)$) such that $u \equiv 0 \mod 2^{i+1}$. The congestion on this edge increases by $|S_i(v)|$ where $v \equiv 2^i \mod{2^{i+1}}$. Applying our inductive hypothesis we get that the total congestion on such an edge is at most $|S_i(v)| + |S_i(u)| = 2ci$, implying the claimed bound on the congestion. 
\end{proof}
}
\fullOnly{\SCutPathProof}
\shortOnly{
}

\gdef\scutPathFig{
\begin{figure}
 \centering
	\includegraphics[scale=0.13]{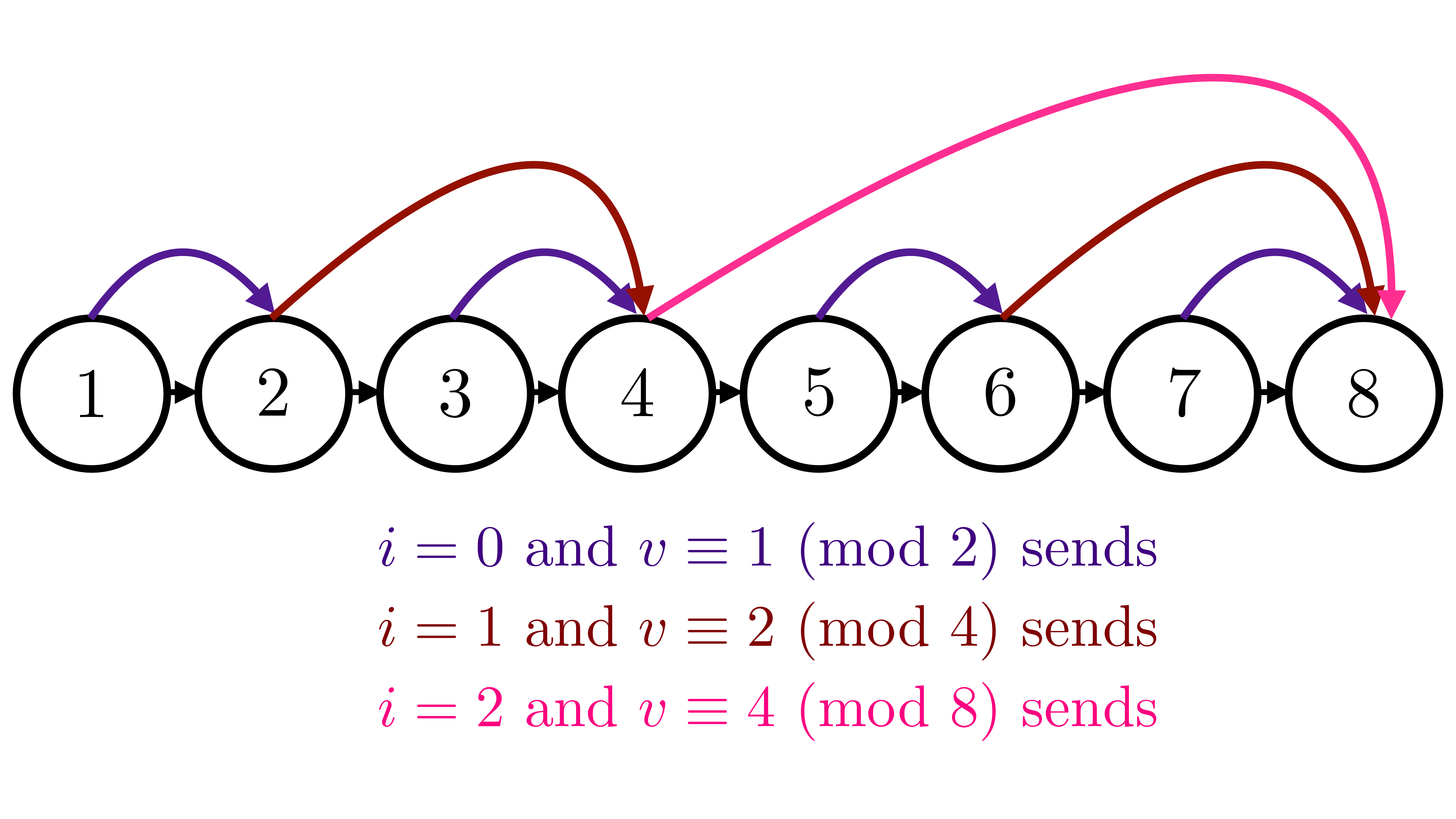}
	\caption{Illustration of \calg{alg:pathSCut}. Source of colored edge gives $v$ that transmits $S_i(v)$ and sink gives $u$ that updates $S_{i+1}(u)$ (assuming no edges broken). Black edges give path edges.}\label{fig:path-shortcut}
\end{figure}
}\fullOnly{\scutPathFig}

We now turn to describing the overall shortcut construction algorithm, \calg{alg:SCutDet}, and analyze the resulting block parameter there. We limit message complexity by only allowing sub-part representatives to send a message requesting that an edge be used in their part's shortcut. As we show, each bottom-up computation yields good shortcuts for a constant fraction of parts. Thus, after each bottom-up computation, we can use our block parameter verification algorithm -- \Cref{lem:verifyb} -- to identify the parts for which our shortcut construction succeeded and freeze the shortcut edges of said parts. The correctness and runtime of \calg{alg:SCutDet} is given by \Cref{lem:ScutDet}.

\gdef\SCutDetPC{
\begin{algorithm}
	\caption{Deterministic shortcut construction.}
	\label{alg:SCutDet}
	\begin{algorithmic}[1]
	\Statex \textbf{Input}: 
	partition of $V$, $(P_i)_{i=1}^N$ with leader $l_i$ known by $v \in P_i$;
	\Statex  \color{white}\textbf{Input}: \color{black} 
	sub-part division
	\Statex  \color{white}\textbf{Input}: \color{black}
	BFS-tree $T$
	\Statex \textbf{Output}: $T$-restricted shortcut with congestion $\tilde{O}(c)$ and block parameter $<3b$
	\State Initially all $P_i$ are \textit{active}.
	\State Compute heavy path decomposition of $T$.
	\For{$j = 1,2,\dots,O(\log n)$}
		\ForAll{$v\in V$}
			\If{$v$ is representative in active part $P_i$}
				\State $S_j(v) \leftarrow \{l_i\}$. 
			\Else
				\State $S_j(v) \leftarrow \emptyset$.
			\EndIf
		\EndFor
	
		\State Set each heavy path with no incoming light edges \textit{active}.
		\For{$\lfloor \log n \rfloor$ repetitions}
			\State Let $S_f$ be the output of $\calg{alg:pathSCut}$ run on all active heavy paths.
			\State For active path sink node $v$ and light edge $(v, u)$ let $S_j(u) = S_{j-1}(u) \cup \bigcup_v S_{f}(v)$. \label{line:activePar}
			\State Set all active paths inactive and all heavy paths with source $u$ as in \Cref{line:activePar} active.
		\EndFor
		\State Set parts with block parameter $< 3b$ \textit{inactive} (see \Cref{lem:verifyb}).
	\EndFor\\
	\Return $\cup_j S_j(v)$ as $v$'s shortcut edges
	\end{algorithmic}
\end{algorithm}}\fullOnly{\SCutDetPC}

\begin{restatable}{lem}{SCutDet}
\label{lem:ScutDet}
	Given:\ partition $(P_i)_{i=1}^N$ where $v \in P_i$ knows leader $l_i \in P_i$; a tree $T$ of depth $D$ which admits a $T$-restricted shortcut of congestion $c$ and block parameter $b$; and a sub-part division, \calg{alg:SCutDet} deterministically computes in $\tilde{O}(b(c + D))$ rounds and $\tilde{O}(m)$ messages, a shortcut with congestion $O(c\log n)$ and block parameter $O(b)$.
\end{restatable}
\gdef\SCutDetProof{
\begin{proof}
	We first prove the runtime and message complexity. As mentioned above, a heavy path decomposition of $T$ can be computed in $O(D)$ rounds using $O(n)$ messages. We now bound the message and round complexity of each iteration. First note that we can inform every path if it has a light edge in $O(D)$ rounds with $O(n)$ messages. Next, running \calg{alg:pathSCut} $O(\log n)$ times -- once on each heavy path -- requires $O(c\log D\log n + D\log n)$ rounds and $O(n \log n)$ messages. Lastly, notice that informing every node in a path that the path is now active requires $O(D)$ rounds using $O(n)$ messages. Lastly, by \Cref{lem:verifyb}, running \calg{alg:pathSCut} is within our stated bounds. Summing over iterations, we conclude that the overall round and message complexities are $\tilde{O}(b(c + D))$ and $\tilde{O}(m)$ respectively.
		
We now prove correctness. Notice that by \Cref{lem:path-shortcut} the number of parts assigned to an edge in any particular iteration is at most $O(c \log D)$ and so the overall congestion on any edge is at most $O(c \log D \log n) = \tilde{O}(c)$. 

We now analyze the block parameter. In particular, we argue that the number of active parts is at least halved in each iteration. Let $A_j$ be the set of active parts in iteration $j$. Let $U_j$ be the set of heavy edges used by $\mathcal{H}$ but broken in iteration $j$ and therefore not assigned to any parts by $S_j$. Each edge in $U_j$ received at least $2c-c=c$ more requests by parts to use it than in $\mathcal{H}$. Thus each edge in $U_j$ receives at least $2c-c=c$ requests from parts in $A_j$. However, each part in $A_j$ can contribute at most $b$ such additional requests to a broken edge, as each block can only send one additional request towards the tree's root. Consequently, we have $|U_j|\leq A_j\frac{b}{2c}$.  Next, we say an active part is \emph{bad} in iteration $j$ if more than $2b$ of its edges of $\mathcal{H}$ are broken in iteration $j$, and \emph{good} in iteration $j$ otherwise. Note that for a good part the number of blocks in the output shortcut is at most $3b = O(b)$. On the other hand, every broken heavy edge used in $\mathcal{H}$ is used at most $c$ times in $\mathcal{H}$. We conclude that the number of bad active parts is at most $|U_j|\frac{c}{2b}$. 
	
	Combining both upper and lower bounds on $|U_j|$, the number of bad parts active parts is at most $A_j/2$ in iteration $j$. Thus, after $O(\log n)$ iterations all parts will be marked inactive, meaning the block parameter in the returned shortcut is at most $3b$.
\end{proof}
}
\fullOnly{\SCutDetProof}
\shortOnly{
\begin{proof}[Proof (Sketch)]
The stated round complexity comes from summing sub-routines. The message complexity follows similarly and requires noting that only the $\tilde{O}(n/D)$ sub-part representatives forward their part ID up at most $D$ nodes when running \Cref{alg:pathSCut}. Lastly, there can only be a constant number of parts with bad block parameter in each iteration. Each extra block we introduce corresponds to a broken edge that is not broken by the original shortcut. However, each such edge also corresponds to some number of parts that tried to use it and since a part can only try to use so many edges in our shortcut we have that not that many extra edges can be broken.
\end{proof}
}

%% file: appendix.tex
\appendix

\section*{Appendix}

\section{Applications of Our PA Algorithms}
\label{sec:appApps}

Here we outline the applications of our round- and message-optimal PA algorithms for multiple distributed graph problems. We start with a discussion of the problems referred to in \Cref{sec:apps}, and then discuss additional applications of our PA algorithm to optimization as well as verification problems, in \Cref{sec:further-apps}.

\subsection{Deferred Proofs of \Cref{sec:apps}}

Here we address how to apply our PA algorithm to solve MST, Approximate Min-Cut and Approximate SSSP, starting with a formal definition of these problems.

\shortOnly{
\probDfnAll
	\probDfnMST
	\probDfnMC
	\probDfnSSSP
\\}

We now restate the round and message complexities we obtain for these problems using our new PA algorithm and discuss the algorithms used to achieve these bounds. As before, recall that since every graph admits a shortcut with $b=1$ and $c = \sqrt{n}$, our algorithms simultaneously achieve worst-case optimal $\tilde{O}(m)$ message complexity and worst-case optimal $\tilde{O}(D  + \sqrt{n})$ round complexity.

\mstcor*
\mstProof

\mincutcor*
\mincutProof

\spcor*
\spProof

\subsection{Further Applications of Our PA Algorithms}\label{sec:further-apps}
Here we mention some further applications of PA, all of which prior work has show to have PA as their round and communication bottleneck. For all of these problems our new PA algorithm of \Cref{thm:partComputation} therefore yields $\tilde{O}(D+\sqrt n)$-round and $\tilde{O}(m)$-message algorithms.

\paragraph{Graph Verification Problems.} In 
\cite{dassarma2012distributed}, Das Sarma et 
al.~provided an extensive list of lower bounds 
for optimization problems (many of which we 
referred to throughout this paper, as their 
lower bounds prove our algorithms' round 
complexity to be optimal). Das Sarma et 
al.~further showed that their lower bounds 
carry over to \emph{verification} problems. For 
these problems, the input is a graph $G$ and a 
subgraph $H$ of $G$ and an algorithm must 
verify whether this subgraph satisfies some 
property, such as whether $H$ is a spanning tree 
or $H$ is a cut (see 
\cite{dassarma2012distributed} for more). 

Das Sarma et al.~also provided algorithms for 
this long list of verification problems, relying heavily on an optimal MST algorithm and 
the following connected component algorithm
of \citet[Algorithm 5]{thurimella1997sub}. Thurimella's algorithm, given a graph $G$ and 
subgraph $H$ as above, outputs a label $\ell(v)$ for each vertex $v\in V(G)$ such that 
$\ell(u)=\ell(v)$ if and only if $u$ and $v$ are in the same connected component of 
$H$. The problem solved by 
Thurimella is easily cast as an instance of PA, 
by having each part elect a leader in a 
connected component in $H$ -- say, a node of 
minimum ID -- and use the leader's ID as a 
label.\footnote{We note that for bipartiteness verification, 
Das Sarma et al.~relied on the algorithm of \cite{thurimella1997sub} also outputting a 
rooted spanning tree of each connected 
component of $H$ with each vertex knowing its 
level in the tree. As our PA algorithm 
maintains such rooted spanning trees, it can 
also be used to solve this verification problem 
within the same bounds.} Without repeating the arguments of 
\citet{dassarma2012distributed}, we note that 
as MST and Thurimella's algorithm require 
$\tilde{O}(D+\sqrt n)$ rounds (based on 
\cite{kutten1995fast}), so do the algorithms of 
Das Sarma et al., and this is tight for these 
verification problems by their lower bounds. 
Our MST and PA algorithms show that for the long list of 
verification problems Das Sarma et al.~considered, optimal 
round complexity does not preclude optimal 
message complexity, as we can attain both 
simultaneously.

\begin{cor}
	All the graph verification problems considered in \cite[Section 8]{dassarma2012distributed} can be solved in an optimal $\tilde{O}(D+\sqrt n)$ rounds and $\tilde{O}(m)$ messages.
\end{cor}

\paragraph{Approximation of Minimum-Weight Connected Dominating Set}
Another application of our PA algorithms follows from the work of \citet{ghaffari2014near}. In that work, Ghaffari shows that the algorithm of \citet{thurimella1997sub}, 
discussed above, can be used to to compute an $O(\log n)$-approximate minimum weight 
connected dominating set (a set of nodes $S$ such that all vertices in $G$ are at distance 
at most one from a node in $S$). In particular, Ghaffari relied on the ability to extend 
Thurimella's algorithm to the case where nodes also have some value $x(v)$ assigned to them so 
that the label of nodes in a connected 
component can be equal to: (A) the list of the 
$k=O(1)$ largest values $x(v)$ in the component, or (B) the sum of values $x(v)$ in the component. Both these applications can be cast as instances of PA. Plugging in our PA algorithm into Ghaffari's algorithm \cite{ghaffari2014near}, we obtain the following.

\begin{cor}
	There exists an $\tilde{O}(D+\sqrt n)$-round, $\tilde{O}(m)$-message algorithm which computes an $O(\log n)$-approximate minimum weight connected dominating set.
\end{cor}

\paragraph{$k$-dominating sets.} Another 
distributed primitive used in multiple 
distributed graph algorithms is the problem of 
computing an $O(n/k)$-node $k$-dominating set. 
That is, a set of nodes $S$ such that every 
node in $G$ is at distance at most $k$ from 
some node in $S$. This problem has found 
applications in distributed algorithm for MST 
\cite{kutten1995fast} and 
$(1+\epsilon)$-approximate eccentricity 
computation \cite{holzer2012optimal}. For this 
problem $\tilde{O}(k)$-round algorithms are known, \cite{kutten1995fast}), including some 
with linear message complexity 
\cite{penso2004distributed}. However, for large 
$k$, i.e. $k\gg \max\{D,\sqrt n\}$, no 
$\tilde{O}(D+\sqrt n)$-round, 
$\tilde{O}(m)$-message algorithm was known. 
Such an algorithm follows immediately from a 
simple generalization of our sub-part division 
algorithm, as follows. As in 
\Cref{alg:subPartDet}, we repeatedly merge 
sub-parts, marking a sub-part as complete when 
it attains some size. Unlike the above 
algorithm, this threshold is chosen to be $k/6$ 
rather than $D$. By arguments which are a 
simple generalization of 
\Cref{alg:subPartDet}'s analysis, this implies 
that the obtained sub-parts have diameter at 
most $k$, and that each sub-part contains at 
least $k/6$ nodes, one of which is its 
representative. This set of representatives 
therefore has cardinality at most $6n/k$ and it 
forms a $k$-dominating set. The one delicate 
point to notice is that now we can solve 
Part-Wise Aggregation using our round- and 
message-optimal algorithm for PA -- unlike in 
\Cref{alg:subPartDet}, which is a sub-routine in our PA algorithm. In particular, even if 
$k\gg \max\{D,\sqrt n\}$, this can be done in $\tilde{O}(D+\sqrt n)$ rounds (and 
$\tilde{O}(m)$ messages). Therefore, as each of the $O(\log n)$ iterations of this algorithm 
can be implemented using some $O(\log^*n)$ many calls to PA (by Lemma \ref{lem:starJoinDet}), together with some local computation, we obtain the following corollary of \Cref{thm:partComputation}.

\begin{cor}
	For any integer $k$, there exists an $\tilde{O}(D+\sqrt n)$-round, $\tilde{O}(m)$-message algorithm which computes a $k$-dominating set of size $O(n/k)$.
\end{cor}

This last corollary bolsters our confidence that the ubiquitous $\tilde{\Theta}(D+\sqrt n)$ bounds 
for distributed graph algorithms can be matched with $\tilde{O}(m)$ messages for an even wider range of problems not discussed in this paper.

\input{puttingAllTogether}

\shortOnly{
	
	\section{Deferred Proofs, Pseudocode and Figures of \Cref{sec:techSet}}
	\label{sec:appKLPA}
	
In this section we give pseudo-code for and the correctness and runtime proofs of our main PA algorithm that takes a sub-part division and shortcut as input. We also do the same for the algorithm of \citet*{haeupler2016low} which we use to aggregate within shortcut blocks and our algorithm for verifying block parameter.
	
	\SPAndBlocksFig	
		

	\KLPAWithSPDAndSC
	\algoIterationsFig
	\solveKLPA*
	\solveKLPAProof
	
	\verifyb*
	\verifybProof
	
	\section{Deferred Proofs and Pseduocode of \Cref{sec:rando}}
	\label{sec:appRand}
		In this section we give our randomized sub-part division and shortcut construction algorithms along with the properties of these algorithms.
	\subPartRandPC
	\SPDivRand*
	\SPDivRandProof
	
	\randSCPC
	\SCutRand*
	\SCutRandProof
	
	
	\section{Deferred Proofs, Pseudocode and Figures of \Cref{sec:det}}\label{sec:appDet}
	Here we provide pseudo-code and the correctness and runtime of our algorithms for deterministic star-joining constructions, deterministic sub-part division construction and deterministic shortcut construction.

	\starJoinPC
	\sJoin*
	\sJoinProof
	
	\detSPPC
	\SPDivDet*
	\SPDivDetProof
	
	\SCutPath*
	\SCutPathProof
	\scutPathFig
	
	\SCPathText
	\detSCPC	
	
	\SCutDetPC
	\SCutDet*
	\SCutDetProof


}

\section{Our Results In Tabular Form}\label{sec:fullResultsTable}
Throughout the paper we state our results in utmost generality by giving our algorithms' 
running time in terms of the optimal block parameter $b$ and congestion $c$. As stated in 
\Cref{thm:partComputation} and its Corollaries \ref{cor:mst}, \ref{cor:minCut} and 
\ref{cor:sp} and \Cref{sec:appApps}, for PA and the wide range of application problems we 
consider, our deterministic algorithms terminate in $\tilde{O}(b(D+c))$ rounds and our 
randomized algorithms terminate in $\tilde{O}(bD+c)$ rounds. To make these bounds 
more concrete, we review some known bounds on the parameters $b$ and $c$ in 
\Cref{table:parameter-bounds}, and then state the implied running times of all our algorithms 
for the above problems in \Cref{table:time-bounds}.\footnote{The two
exceptions to this rule are our 
$L^{\epsilon}$-approximate SSSP and $(1+\epsilon)$-approximate Min-Cut, for which 
the (randomized) bounds hold as stated in the table only for fixed (or polylogarithmic) $\epsilon$.}

\begin{table*}[h]
	\centering
	\begin{tabular}{|c|c|c|c|c|c|c|c|}
	        \noalign{\vskip .5mm}  

		\hline
		\multirow{2}{*}{} & General  & Planar  & Genus $g$  &  Treewidth $t$  & Pathwidth $p$  & Minor Free  \\
		{} &  \cite{ghaffari2016distributed} &  \cite{ghaffari2016distributed} & \cite{haeupler2016low} &   \cite{haeupler2016near} & \cite{haeupler2016near} & \cite{haeupler2018minor}  \\
		\hline
		\noalign{\vskip .5mm}  
		$b$  & $\bigstrut1$  & $O(\log D)$ & $O(\sqrt{g})$  & $O(t)$  & $p$ & $\tilde{O}(D)$ \\
		\hline   
		\noalign{\vskip .5mm}   
		$c$  & $\sqrt{n}$  & $\tilde{O}(D)$ &  $\tilde{O}(\sqrt{g}D)$ & $\tilde{O}(t)$& $p$ & $\tilde{O}(D)$ \\
		\hline
	\end{tabular}
	\captionsetup{justification=centering}
	\caption{Known bounds on block parameter, $b$, and congestion, $c$.}
	\label{table:parameter-bounds}
\end{table*}

\begin{table*}[h]
	\centering
	\begin{tabular}{|c|c|c|c|c|c|c|}
		\hline
		& General & Planar & Genus $g$ &  Treewidth $t$ & Pathwidth $p$ & Minor Free \\
		\hline

		Deterministic & \bigstrut$\tilde{O}(D+\sqrt{n})$  & $\tilde{O}(D)$ & $\tilde{O}(gD)$ & $\tilde{O}(tD + t^2)$ & $\tilde{O}(pD + p^2)$ & $\tilde{O}(D^2)$ \\
		\hline

		Randomized & \bigstrut$\tilde{O}(D+\sqrt{n})$  & $\tilde{O}(D)$ & $\tilde{O}(\sqrt{g}D)$  & $\tilde{O}(tD)$ & $\tilde{O}(pD)$ & $\tilde{O}(D^2)$ \\
		\hline
	\end{tabular}
	\captionsetup{justification=centering}
	\caption{Summary of running times of our algorithms.}
	\label{table:time-bounds}
\end{table*}

Reviewing \Cref{table:time-bounds}, we note that for all problems considered, a matching worst case round lower bound of $\tilde{\Omega}(D+\sqrt{n})$ is given by \citet{dassarma2012distributed}, while a trivial lower bound of $\Omega(D)$ holds for these problems for all graphs. Our algorithms match the worst case bounds and the $\Omega(D)$ lower bound (up to polylog terms) for any constant genus, treewidth and pathwidth, all while requiring only $\tilde{O}(m)$ messages. The exact optimal dependence on the parameters $g,t$ and $p$ remains an open question.

%% file: puttingAllTogether.tex
\section{Dispensing with Known Leader Assumption}\label{sec:PASolns}
Throughout this paper we have assumed that parts always know a leader. That is for every part $P_i$ every node $v \in P_i$ knows the ID of some leader $l_i \in P_i$. We solved PA assuming that this holds. We now show that this assumption can dispensed with. In particular, we demonstrate that an algorithm that solves PA with the assumption of a known leader for each part can be converted into one that makes no such assumption with only logarithmic overhead in round and message complexities. The conversion is deterministic and so it demonstrates that a known leader is not required for either or deterministic or our randomized results. 

Our algorithm is \Cref{alg:convLeaderSolution} and works as follows: start with the singleton partition where every node is its own leader; repeatedly coarsen this partition $O(\log n)$ times until it matches the input PA partition by applying our PA solution that assumes that a leader is known to merge the stars given by a star joining. At each step in the coarsening we maintain the invariant that every part knows a leader and so in the end we need only solve PA with a known leader which we can do by assumption.

\gdef\solvePAPC{	
\begin{algorithm}
	\caption{PA without leaders.}
	\label{alg:convLeaderSolution}
	\begin{algorithmic}[1]
	\Statex \textbf{Input}: PA instance with parts $(P_i)_{i=1}^N$; 
	\Statex  \color{white}\textbf{Input}: \color{black} PA algorithm, $\mathcal{A}$, that assumes every part knows a leader.
	\Statex \textbf{Output}: a solution to the input PA problem.
	\ForAll{$i \in |V|$}
		\State Set $P_i' \leftarrow \{ i\}$ and $l_i \leftarrow i$. \Comment{Each $P_i'$ maintains a leader $l_i$, initially set to $i$}
	\EndFor
	\For{$O(\log n)$ rounds}
		\ForAll{part $P_i'$} 
			\State Pick some $e_i = (u, v) \in P_i' \times (V\setminus P_i')$ by running $\mathcal{A}$ where $u, v$ are in the same $P_i$.
		\EndFor
		\State Compute a star joining with \calg{alg:starJoinDet} over the $P_i'$s using edges $\{e_i\}$.
		\ForAll{part $P_i'$ which joined $P_j'$ in the star joining} 
			\State Inform each $v \in P_i'$ that their leader is now $l_j$ by running $\mathcal{A}$.
			\State Merge $P_i'$ into $P_j'$.
		\EndFor
	\EndFor
	\State Run $\mathcal{A}$ on the PA instance consisting of the $P_i'$s, each with a known leader $l_i$.
	\end{algorithmic}
\end{algorithm}
}\solvePAPC

\begin{restatable}{lem}{test}
Given a PA instance where no leaders are known and a PA algorithm, $\mathcal{A}$, that assumes leaders are known using $R$ rounds and $M$ messages, \Cref{alg:convLeaderSolution} solves the PA instance with no leaders in $\tilde{O}(R)$ rounds and $\tilde{O}(M)$ messages.
\end{restatable}

\begin{proof}
We first prove round and message complexities. Our algorithm runs $\mathcal{A}$ to solve PA with a known leader and \calg{alg:starJoinDet} to compute a star-joining logarithmically many times. The latter consists of $O(\log ^*n)$ calls to $\mathcal{A}$. Thus, the stated round and message complexities follow trivially.

We now argue correctness. Each round a constant fraction of the $P_j'$s participating in the algorithm get to merge by definition of a star joining and so $O(\log n)$ repetitions are sufficient to coarsen every $P_j'$ to a $P_i$. Moreover, $P_1', \ldots, P_{N'}'$ is valid input to $\mathcal{A}$ since we maintain the invariant that every node in a $P_j'$ has an elected leader. At the end of this coarsening our PA instance now has elected leaders. By the correctness of $\mathcal{A}$ our algorithm is correct.
\end{proof}